\documentclass{ieeeconf}                                        
                                                          
\usepackage{amsmath}
\usepackage{amssymb}
\usepackage{amsfonts}

\usepackage{amsthm}
\usepackage{bm}
\usepackage{graphicx}
\usepackage{url}
\usepackage{booktabs}
\usepackage{color}
\usepackage{xcolor}
\usepackage{soul}
\usepackage{colortbl}
\usepackage{siunitx}
\usepackage[caption=false]{subfig}
\usepackage{cases}
\usepackage[noadjust]{cite}
\usepackage{algorithm}
\usepackage[noend]{algpseudocode}
\usepackage{mathtools}
\usepackage{textcomp}
\def\BibTeX{{\rm B\kern-.05em{\sc i\kern-.025em b}\kern-.08em
    T\kern-.1667em\lower.7ex\hbox{E}\kern-.125emX}}
\usepackage[colorlinks=true, allcolors=blue]{hyperref}

\newcommand{\des}{\mathrm{des}}
\newcommand{\rect}{\mathrm{rect}}
\newcommand{\nom}{\mathrm{nom}}

\newcommand{\R}{\mathbb{R}}

\newcommand{\Rn}{\mathbb{R}^n}
\newcommand{\Rm}{\mathbb{R}^m}
\newcommand{\Rnm}{\mathbb{R}^{n \times m}}

\newcommand{\Sc}{\mathcal{S}}

\newcommand{\x}{\mathbf{x}}
\newcommand{\y}{\mathbf{y}}

\newcommand{\X}{\mathbf{X}}

\newcommand{\uu}{\mathbf{u}}
\newcommand{\unom}{\uu_{\mathrm{nom}}}
\newcommand{\urect}{\uu_{\mathrm{rect}}}
\newcommand{\xquery}{\x_q}
\newcommand{\xnoisy}{\underline{\x}_q}

\newcommand{\pipex}{\bigg|_{ \xquery}}
\newcommand{\pipemu}{\bigg|_{ \bm{\mu}}}
\newcommand{\pipemuk}{\bigg|_{ \mu_k}}
\newcommand{\Kbar}{\mathbf{ \overline{K} }}
\newcommand{\KbarInv}{\mathbf{ \overline{K} \hspace{0.05cm} }^{-1}}

\newcommand{\hgp}{ h_{\mathrm{gp}} }

\newcommand{\hb}{ h_{\mathrm{b}} }

\newcommand{\hu}{ h_{\mathrm{u}} }

\newcommand{\dhgpdx}{\frac{\partial \hgp (\x) }{\partial \x}}

\newcommand{\dmudx}{\frac{\partial \mu (\x) }{\partial \x}}
\newcommand{\dvardx}{\frac{\partial \sigma^2 (\x) }{\partial \x}}
\newcommand{\dkdx}{\frac{\partial \mathbf{k} (\x) }{\partial \x}}
\newcommand{\dkidx}{\frac{\partial \mathbf{k}_{(i)} (\x) }{\partial \x}}
\newcommand{\ddkiddx}{\frac{\partial^2 \mathbf{k}_{(i)} (\x) }{\partial \x^2}}

\newcommand{\dqdx}{\frac{\partial \mathbf{q} (\x) }{\partial \x}}

\newtheorem{theorem}{Theorem}
\newtheorem{remark}{\textit{Remark}}

\theoremstyle{definition}
\newtheorem{definition}{Definition}

\newtheorem{proposition}{Proposition}

\theoremstyle{problem}
\newtheorem{problem}{Problem}

\theoremstyle{assumption}
\newtheorem{assumption}{Assumption}

\definecolor{1}{RGB}{255, 240, 0}
\newtheoremstyle{colth}%
{}{}%
{\itshape}{}%
{}{}%
{ }%
{\colorbox{1}{\bf{\thmname{#1}\thmnumber{ #2}} \thmnote{(#3)}.}}

\theoremstyle{colth}

\DeclareMathOperator*{\argmin}{arg\,min}

\newcommand\norm[1]{\left\lVert #1 \right\rVert}

\newlength\mylen

\let\oldnl\nl
\newcommand{\nonl}{\renewcommand{\nl}{\let\nl\oldnl}}

\definecolor{green}{rgb}{0,1,0}
\definecolor{lightblue}{rgb}{0,0,0.8} 
\definecolor{red}{rgb}{1,0,0}

\newcommand{\greensquare}{{\color{green}$\blacksquare$}}
\newcommand{\bluecirc}{{\color{lightblue}$\circ$}}
\newcommand{\reddisc}{{\color{red}$\bullet$}}

\begin{document}

\title{Gaussian Control Barrier Functions : A Non-Parametric Paradigm to Safety}

\author{Mouhyemen Khan, Tatsuya Ibuki$^{\dag}$, Abhijit Chatterjee
\thanks{Mouhyemen Khan and Abhijit Chatterjee are with the School of Electrical and Computer Engineering, Georgia Institute of Technology, Atlanta, USA: {\tt \{mouhyemen.khan, abhijit.chatterjee\}@gatech.edu } \linebreak
$\dag$ Tatsuya Ibuki is with the Department of Electronics and Bioinformatics, School of Science and Technology, Meiji University, Kanagawa 214-8571, Japan: {\tt ibuki@meiji.ac.jp }
}
}

\maketitle

\begin{abstract}
Inspired by the success of control barrier functions (CBFs) in addressing safety, and the rise of data-driven techniques for modeling functions, we propose a non-parametric approach for online synthesis of CBFs using Gaussian Processes (GPs).
Mathematical constructs such as CBFs have achieved safety by designing a candidate function a priori. However, designing such a candidate function can be challenging. A practical example of such a setting would be to design a CBF in a disaster recovery scenario where safe and navigable regions need to be determined. The decision boundary for safety in such an example is unknown and cannot be designed a priori. 
In our approach, we work with \textit{safety samples or observations} to construct the CBF online by assuming a flexible GP prior on these samples, and term our formulation as a \textit{Gaussian CBF}. 
GPs have favorable properties, in addition to being non-parametric, such as analytical tractability and robust uncertainty estimation. This allows realizing the posterior components with high safety guarantees by incorporating variance estimation, while also computing associated partial derivatives in closed-form to achieve safe control. 
Moreover, the synthesized safety function from our approach allows changing the corresponding safe set arbitrarily based on the data, thus allowing non-convex safe sets.
We validate our approach experimentally on a quadrotor by demonstrating safe control for fixed but arbitrary safe sets and collision avoidance where the safe set is constructed online. 
Finally, we juxtapose Gaussian CBFs with regular CBFs in the presence of noisy states to highlight its flexibility and robustness to noise.
The experiment video can be seen at: \url{https://youtu.be/HX6uokvCiGk}.
\end{abstract}


\maketitle

\section{INTRODUCTION}
With the rise of autonomous systems, assuring their safety is of paramount importance. For instance, an autonomous drone should not crash during its mission, or a self-driving vehicle should not collide with other vehicles. Formulating constraints for these applications to ensure safety is a difficult task. These intelligent systems often use data-driven solutions for learning and adapting online. However, it is unclear how to principally and efficiently encode data into the paradigm of safety.
A popular approach to ensuring safety of dynamical systems leverages set theoretic ideas. To be more specific, the theory of controlled set invariance is employed where a system is defined to be safe if (a subset of) its states remain within a prescribed set \cite{controls_sets_Blanchini2008}. This forms the basis of control barrier functions (CBFs) which have been successfully demonstrated on many safety-critical applications \cite{cbf_Ames2017}, \cite{cbf_safe_swarm_quadrotors_Li2017}, \cite{cbf_biped_Hsu2015}. To incorporate safety using CBFs, we need two items: a candidate function or certificate satisfying the required relative degree and a nominal model of the system dynamics. The candidate function defines a superlevel set in which the system remains forward invariant subject to certain constraints. Traditionally, these safety functions have been hand-designed. However, depending on the application, designing such a candidate CBF is not straight forward in many practical settings. To highlight the difficulty of a hand-designed solution, consider the example of a disaster recovery scenario. Designing a safety function manually for this scenario requires great effort and intuition to find an appropriate form. This can compromise system safety if designed incorrectly. We believe designing the safety candidate function based on sampled data in a principle and efficient manner will alleviate many of these concerns. To this end, we propose working with \textit{safety samples, which encode a safety metric of interest, to synthesize the safety function in a data-driven manner using Gaussian Processes (GPs) and define it as a Gaussian CBF} (see Figure \ref{fig:gcbf_vs_cbf_contours}).

CBFs achieve safe control using barrier certificates \cite{cbf_barrier_certificates_prajna2007, cbf_constructive_wieland2007}. The system's safety is encoded using these safety barrier certificates (or safe sets) with the aid of a smooth function satisfying certain properties. These functions can then be combined with quadratic programs (QPs) to achieve safety constrained control \cite{cbf_Ames2017, cbf_safe_swarm_groundbots_Li2015, cbf_safe_swarm_quadrotors_Li2017, so3_conic_cbf_Tatsuya2020}. Certificates based on Lyapunov and barrier functions were combined to demonstrate stable and safe constrained control \cite{cbf_2dquad_Wu2016, cbf_3dquad_Wu2016}. These approaches incorporate deterministic CBFs defined a priori without incorporating any form of data for altering the CBF. 

Learning based methods for addressing safety have been investigated previously. For an uncertain nonlinear system, the region of attraction is learned by using non-parametric GPs and Bayesian optimization (BO) to estimate and expand the safe set in \cite{safe_learning_roa_felix2016} . BO was also used in safety-critical systems such as quadruped, snake, and quadrotor for improving system performance while ensuring safety \cite{bo_gait_lizotte2007}, \cite{bo_snake_matthew2011}, \cite{bo_quadrotor_berkenkamp2016}. 
Unfortunately, the high run-time complexity of BO limits its applicability to evolving tasks or changing environmental conditions. For systems with polynomial dynamics, an optimization routine can be set up as a convex semi-definite problem using sum-of-squares (SoS) technique to search for a valid safety certifcate \cite{opt_sos_ahmadi2016, opt_sos_ahmadi2019}. However, SoS methods also scale poorly with high dimensions, similar to BO, and are limited to polynomial system dynamics.

In the context of CBFs, data-driven techniques are actively pursued. Support vector machines (SVMs) were used in \cite{cbf_supervisedml_Srinivasan2020} to parameterize CBFs with the help of sensor measurements. Carefully designed weights are required in \cite{cbf_supervisedml_Srinivasan2020} for the SVM classifier to work and the study is confined to simulation results. Data in the form of expert demonstrations was used in \cite{cbf_expertdata_Robey2020}, \cite{cbf_robust_hybrid_Robey2021} to generate CBFs. However, in many applications, having access to expert demonstrations is not always feasible. While both the papers \cite{cbf_supervisedml_Srinivasan2020, cbf_expertdata_Robey2020} empirically verify their findings, neither provide hardware experimental validation of their methods. Adaptive CBFs were formulated to handle time-varying control bounds and noise in the system dynamics \cite{cbf_adaptive_xiao2021}. CBFs have been combined with model predictive control methods for safe motion and path planning \cite{cbf_mpc_polytopes_thirugnanam2022}, \cite{cbf_mpc_racing_he2022}. For stochastic dynamical systems, stochastic CBFs are developed \cite{cbf_stochastic_clark2021}. These approaches use parametric CBFs while accounting for uncertainty in the system dynamics. The authors in \cite{cbf_robust_HJ_choi2021} use a value-function approach by combining Hamilton-Jacobi (HJ) with CBFs to maximize the safe set using viability kernels. However, the HJ based CBF is limited to low-dimensional systems.

Neural certificates, which use neural networks to construct safety barrier certificates, have been demonstrated in \cite{neural_formal_abate2020}, \cite{neural_synthesis_zhao2020}, \cite{neural_convex_tsukamoto2020}, \cite{neural_lyapunov_gaby2021}. These neural certificates provide a data-driven approach to learning-based controllers and provide formal proofs of correctness. A second-order cone program was formulated in \cite{cbf_gaussian_castaneda2021}, with GPs used for modeling the control input and dynamic model uncertainty learned in an episodic manner. However, all these studies have been confined to simulation experiments and are limited to offline training which limits their applicability in many practical online settings. The work in \cite{cbf_gaussian_castaneda2021} was limited to a particular type of kernel to satisfy affineness properties. We differ from \cite{cbf_gaussian_castaneda2021} in modeling the safety candidate function with GPs instead of the underlying system dynamics and are not limited to one type of kernel parameterization.
In our previous work \cite{cbf_gaussian_variance_Khan2021}, safety uncertainty was introduced in CBFs by augmenting the GP posterior variance with an existing CBF using data. We experimentally validated augmentation of safety uncertainty to a given CBF in hardware. However, this required a parametric CBF as the underlying safety function. Moreover, the safe set expansion in \cite{cbf_gaussian_variance_Khan2021} was limited only to convex safe sets. \textit{In this research, we consider a fully non-parametric formulation for synthesizing the safety function without requiring any parametric CBF candidate function. }
We also provide theoretical guarantees in this study.

\begin{figure}[!t]
\centering
\vspace{0.2cm}
\includegraphics[width=1\linewidth]{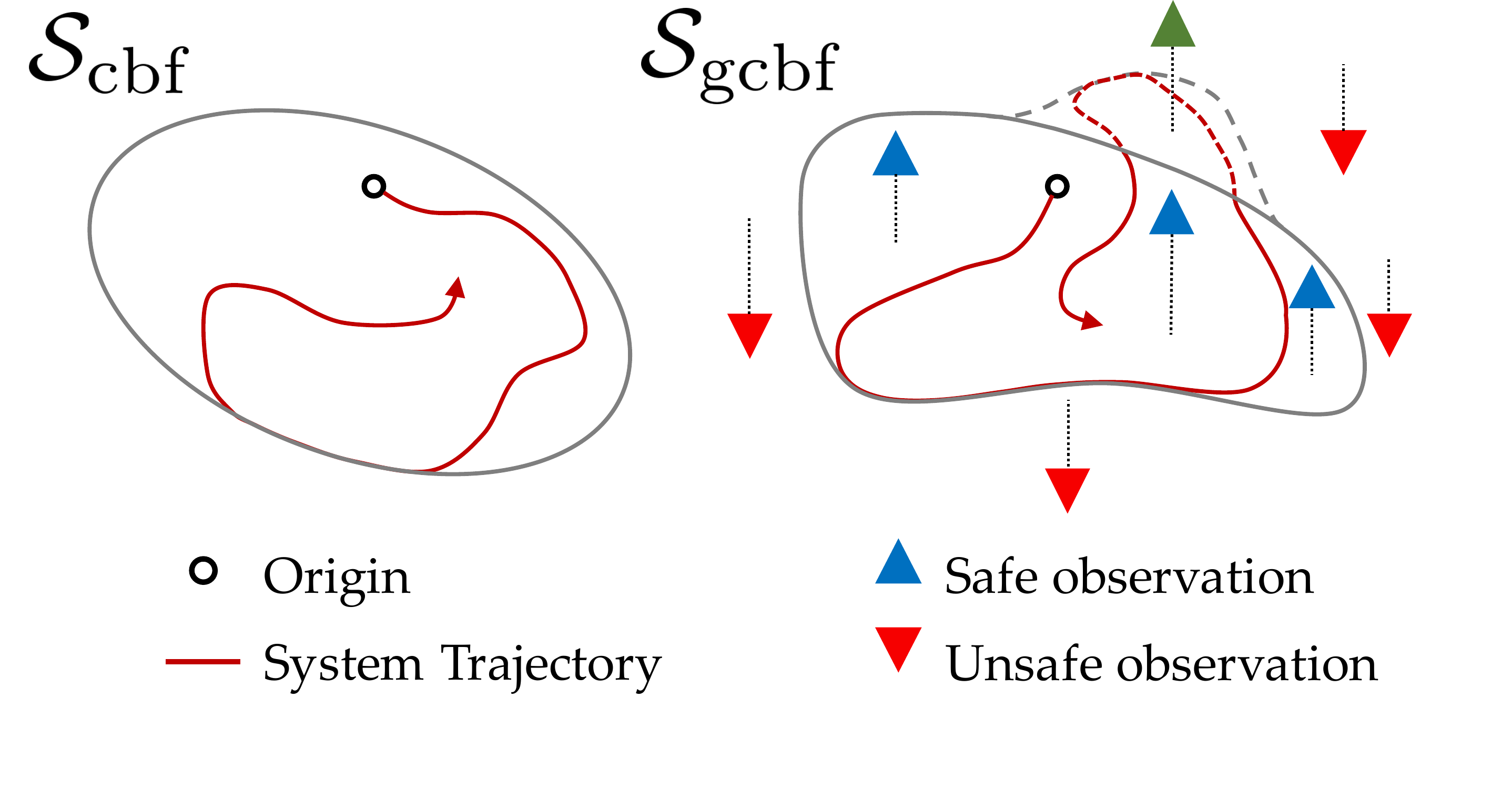}
\caption{The Gaussian CBF is non-parametric and relies on data to produce the safe sets. The $0$-level contour sets are shown for a traditional CBF (left) and Gaussian CBF (right) with safe sets $\Sc_{\mathrm{cbf}}$ and $\Sc_{\mathrm{gcbf}}$ respectively. As a new observation is received (green upper triangle), the Gaussian CBF can change the safe set based on the data in a non-convex fashion.}
\label{fig:gcbf_vs_cbf_contours}
\end{figure}

\noindent Our main \textbf{contributions} are the following. 

\begin{enumerate}
\item We present a novel approach for synthesizing CBFs in a data-driven non-parametric manner using GPs. This is achieved using safety samples as opposed to the prevailing use of CBFs which requires a function.

\item We construct Gaussian CBFs to design safe sets based on the data. These sets are not confined to convex safe sets. GPs provide favorable properties such as analytical tractability and uncertainty estimation which are key enablers in finding closed-form safety function and associated Lie derivatives with high guarantees.

\item We formulate Gaussian CBFs for safe control in the presence of noise for both the safety samples (observations to GPs) and the system states (inputs to the GPs).

\item We validate Gaussian CBFs in hardware using a quadrotor for three case studies: (i) safe control for fixed but arbitrary safe sets, (ii) online obstacle avoidance with an evolving safe set, and (iii) juxtaposing Gaussian CBFs with regular CBFs for safe control in the presence of noisy system states. 

\end{enumerate}

To the best of our knowledge, we believe this is the first work that fully synthesizes a CBF in a non-parametric data-driven manner online using GPs and validate all findings in hardware. The outline of the paper is as follows. Mathematical preliminaries are covered in Section \ref{sec:background} with the problem statement discussed in Section \ref{sec:problem}. We present our proposed methodology in \ref{sec:methodology}. The quadrotor application testcase is explained in Section \ref{sec:application}. We discuss experimental results in Section \ref{sec:experiments} followed by conclusion in Section \ref{sec:conclusion}.

\section{BACKGROUND PRELIMINARIES}\label{sec:background}
We tackle the problem of encoding safety for a dynamical system using a data-driven approach. We leverage theoretical properties of CBFs and GPs to encode safety in a probabilistic manner. We first review key results on CBFs and GPs and refer the reader to \cite{cbf_Ames2017, cbf_exponential_Nguyen2016, gp_textbook_Rasmussen2003} for more details.

\subsection{Control Barrier Function}\label{subsec:cbf}
Consider a general control affine dynamical system,
\begin{align}\label{eq:nonlinear_affine_system}
\dot{\x} &= f(\x) + g(\x) \uu,
\end{align}
where $\x(t) \in \Rn$ is the state and $\uu(t) \in \Rm$ is the control input. The drift vector field, $f : \Rn \rightarrow \Rn$, and control matrix field, $g : \Rn \rightarrow \Rnm$, are assumed to be locally Lipschitz continuous. For reasons of brevity, we omit the implicit dependence on time for the state and control input unless otherwise stated. Let safety for (\ref{eq:nonlinear_affine_system}) be encoded as the superlevel set $\Sc$ of a smooth function $h : \Rn \rightarrow \R $ as follows,
\begin{align}\label{eq:safeset}
\mathcal{S} = \{ \x \in \Rn \ | \ h(\x) \geq 0 \}.
\end{align}

\begin{definition}[Control Barrier Function \cite{cbf_Ames2017}]\label{def:cbf}
\textit{The function $h(\x) : \Rn \rightarrow \R$ is defined as a control barrier function (CBF), if there exists an extended class-$\kappa$ function $\alpha$ ($\alpha(0) = 0$ and strictly increasing) such that for any $\x \in \Sc$,
\begin{align}\label{eq:cbf_inequality}
	\sup\limits_{\uu \in \Rm }  L_f h (\x) + L_g h (\x) \uu + \alpha( h (\x)) \geq 0,
\end{align}
}
\end{definition}
\noindent where $L_f h (\x) = \frac{\partial h }{\partial \x}f(\x)$ and $L_g h (\x) = \frac{\partial h }{\partial \x}g(\x)$ are the Lie derivatives of $h(\x)$ along $f(\x)$ and $g(\x)$ respectively.

\begin{theorem}[Safety Condition \cite{cbf_Ames2017}]\label{thm:cbf}
\textit{Given a system (\ref{eq:nonlinear_affine_system}), with safe set $\Sc \subset \Rn$ (\ref{eq:safeset}), and a smooth CBF $h(\x) : \Rn \rightarrow \R$ (\ref{eq:cbf_inequality}), any Lipschitz continuous controller $\uu \in \Rm$, chosen from $\mathrm{K}_{\mathrm{cbf}} = \{ \uu \in \Rm \ | \  L_fh(\x) + L_gh(\x) \uu + \alpha(h(\x)) \geq 0 \}$ for any $\x \in \Rn$, renders the set $\Sc$ forward invariant for (\ref{eq:nonlinear_affine_system}).
}
\end{theorem}

As seen from Theorem \ref{thm:cbf}, CBFs are limited to systems with relative degree one, $\rho = 1$. For systems with $\rho > 1$, we look at an extension of CBFs called Exponential CBFs \cite{cbf_exponential_Nguyen2016, cbf_highorder_xiao2019}.

\begin{definition}[Exponential Control Barrier Function \cite{cbf_exponential_Nguyen2016}]\label{def:ecbf}
\textit{The smooth function $h(\x) : \Rn \rightarrow \R$, with relative degree $\rho$, is defined as an exponential control barrier function (ECBF), if there exists $\mathcal{K} \in \R^{\rho}$ such that for any $\x \in \Rn$,}
\begin{align*} 
	\sup\limits_{\uu \in \Rm } L_f^{\rho} h(\x) + L_gL_f^{\rho-1}h(\x)\uu + \mathcal{K}^{\top} \mathcal{H} \geq 0,
\end{align*}
\end{definition}
\noindent where $\mathcal{H} = [h(\x), L_fh(\x), ... , L_f^{(\rho-1)} h(\x)]^{\top} \in \R^\rho$ is the Lie derivative vector for $h(\x)$, and $\mathcal{K} = [k_0, k_1, ..., k_{\rho-1}]^{\top} \in \R^\rho$ is the coefficient gain vector for $\mathcal{H}$. $\mathcal{K}$ can be determined using linear control methods such as pole placement. We refer the reader to \cite{cbf_exponential_Nguyen2016} for proofs of ECBF forward invariance.

\subsection{Gaussian Process Regression}\label{subsec:gp}
GPs are a popular choice in machine learning for nonparametric regression which rely on kernels. Kernels furnish a notion of similarity between pairs of input points, $\x_i, \x_j \in \Rn$. However, any arbitrary function of input pairs will not constitute a valid kernel. To be a valid kernel, it should satisfy positive semidefiniteness, see \cite{gp_textbook_Rasmussen2003}.
A popular choice of the kernel function is the squared exponential (SE) kernel,
\begin{align}\label{eq:gaussian_kernel}
k(\x_i,\x_j) = \sigma_f^2 \exp \bigg( \hspace{-0.15cm} - \frac{ ( \x_i - \x_j )^{\top} \mathbf{L}^{-2} ( \x_i - \x_j) } {2}  \bigg) \hspace{-0.1cm} + \hspace{-0.05cm} \delta_{ij} \sigma_y^2 ,
\end{align}
where $\delta_{ij} = 1$ if $i=j$ and $0$ otherwise, $\mathbf{l} \in \Rn$ is the characteristic length scale, with $\mathbf{L} = \mathrm{diag}(\mathbf{l}) \in \R^{n \times n}$. 
The signal scale and observation noise are given by $\sigma_f^2 \in \R$ and $\sigma_y^2 \in \R$ respectively. Together, these free parameters constitute the SE kernel's hyperparameters, $\Theta = \{ \mathbf{L}, \sigma_f^2, \sigma_y^2 \}$. 


We are interested in constructing a safety function for which we assume to have noisy scalar observations.  Given a set of $N$ data points, with input vectors $\x \in \Rn$, and scalar targets $y \in \R$, we compose the dataset $\mathcal{D}_N = \{ \X_N, \y_N \}$, where $\X_N = \{\x_i\}_{i=1}^{N}$ and $\y_N = \{y_i\}_{i=1}^N$. GPs can compute the posterior mean and variance for an arbitrary deterministic query point $ \xquery \in \Rn$, by conditioning on previous measurements. We will later investigate how to handle the case when the query point $\x_q$ is noisy. The posterior mean $\mu \in \R$ and variance $\sigma^2 \in \R$ are given by \cite{gp_textbook_Rasmussen2003},
\begin{align}
\mu( \xquery) &= \mathbf{k}( \xquery)^{\top} \ \KbarInv \y_N, \label{eq:gp_mean} \\
\sigma^2( \xquery) &= k ( \xquery,  \xquery) - \mathbf{k}( \xquery)^{\top} \ \KbarInv \mathbf{k}( \xquery) \label{eq:gp_var},
\end{align}
where $\mathbf{k}(\xquery) = \big[k(\x_1,  \xquery), \ldots , k(\x_N,  \xquery) \big]^{\top} \in \mathbb{R}^{N}$ is the covariance vector between $\X_N$ and $ \xquery$, $\mathbf{\overline{K}} \in \R^{N \times N}$, with entries $[ \bar{k} ]_{(i,j)} = k(\x_i, \x_j), \ i, j \in \{1, \ldots, N\}$, is the covariance matrix between pairs of input points in $\X_N$, and $k( \xquery,  \xquery) \in \R$ is the prior covariance. The SE kernel is infinitely differentiable, and hence, it is infinitely mean-square (MS) differentiable. This allows for a flexible parameterization of the safety function and its Lie derivatives. We use the SE kernel to develop Gaussian CBFs while noting that any valid kernel can be used.

\section{PROBLEM STATEMENT}\label{sec:problem}
Consider a control affine system (\ref{eq:nonlinear_affine_system}) is given, with access to its states $\x$, and scalar noisy observations $y$, that represents a metric for safety. The metric for safety cannot be generalized and therefore is very problem dependent. A distance sensor's readings for obstacle avoidance can be used as a metric for safety or a temperature sensor's readings for determining thermally acceptable regions to traverse. In a similar vein, a LIDAR scan creating $\mathrm{3D}$ point cloud information can be used to detect environmental hazards or a computer vision algorithm providing the decision boundary for safe regions of interest. In all these examples, we can easily sample from the data based on domain knowledge to construct a target metric for safety. \textit{This provides us with the means to construct a valid safety certificate using the data as opposed to hand-designing a safety function which could be limited and requires manual effort along with good domain knowledge intuition.}

\begin{remark}
We assume there is a high-level planner or observer, e.g., sensors or computer vision algorithms, providing the necessary data observations. We acknowledge some feature engineering or data sampling may be involved which is very common in practice. These observations represent the \textit{safety sample} candidates in our problem setting.
\end{remark}


 Our objective is to \textit{synthesize a safety function $\hgp(\x)$ in a non-parametric manner from measurements of the system states and safety samples or observations online and ensure that (\ref{eq:nonlinear_affine_system}) remains safe.} Data-based methods are ultimately approximations and hence, it is desirable to account for any uncertainty in the estimation of the safety function. This leads to the following candidate function,
\begin{align}\label{eq:problem_general_gcbf}
\underbrace{{\hgp(\x(t))}}_{\text{overall safety}} := \underbrace{{\hb(\x(t)  ; \Theta  )}}_{\text{safety belief}} \ - \underbrace{\hu(\x(t) ; \Theta )}_{\text{safety uncertainty}}.
\end{align}

The system's overall safety is given by $\hgp(\x)$ which has two components; a belief in safety given by $\hb(\x)$ and an associated uncertainty given by $\hu(\x)$. Intuitively, the safety belief represents our best estimation of system safety and safety uncertainty represents the uncertainty in the estimation. Ideally, if there is no uncertainty, then the safety belief will perfectly match the overall final safety. Additionally, there are hyperparameters $\Theta$ that can alter the relative notion of safety belief and uncertainty.

\begin{problem} 
Given system \eqref{eq:nonlinear_affine_system} and online (noisy) measurements of the state $\xquery$, synthesize $\hgp(\x)$ with a safety belief and associated uncertainty, conditioned on past states and observations in the dataset given by: $\mathcal{D}_N = \{ \X_N, \y_N \}$, where $\X_N = \{\x_i\}_{i=1}^{N}$ and $\y_N = \{y_i\}_{i=1}^N$, such that system (\ref{eq:nonlinear_affine_system}) is safe.
\end{problem}

To ensure the system remains safe, we need to rectify a given nominal control input $\uu_{\mathrm{nom}}$ to its rectified form $\uu_{\mathrm{rec}}$ which is then applied to (\ref{eq:nonlinear_affine_system}). This is done by making sure the Lie derivatives of the corresponding candidate safety function satisfies the inequality (\ref{eq:cbf_inequality}).


\begin{problem}
Given system \eqref{eq:nonlinear_affine_system}, synthesized $\hgp(\x)$ with safe set $\Sc$, and a nominal control input $\uu_{\mathrm{nom}}$, design the rectified control input $\uu_{\mathrm{rect}}$ such that  system (\ref{eq:nonlinear_affine_system}) is safe.
\end{problem}

Note that designing the control objective for a data-driven based CBF construction is particularly challenging. A non-parametric approach is adopted where the data is fully exploited to construct the safe sets and safety function hypothesis. This compounds the problem of computing the Lie derivatives since time derivatives are computed on $\hgp(\x)$ in order to satisfy the forward invariance properties for CBFs. It is an ill-posed problem to compute the time derivative of an unknown entity i.e., the system's safety belief and uncertainty, without making prior assumptions. For addressing these challenges, a kernel representation is used for constructing the safety belief and uncertainty.

\section{PROPOSED METHODOLOGY}\label{sec:methodology}
In this section, we present our proposed approach, where GPs are used for synthesizing the safety function. A key advantage of GPs over other models such as neural networks, radial basis functions or polynominal chaos, lies in its bayesian non-parametric design. \textit{By allowing a flexible prior over functions, GPs give a probabilistic workflow that gives robust posterior estimates in analytical form. This enables a flexible realization for our safety function as well as computing the associated Lie derivatives.} The resulting architecture for our framework is shown in Figure \ref{fig:gcbf_architecture}. 

\subsection{Gaussian Control Barrier Function}\label{subsec:gcbf}
A GP prior is placed on the desired candidate safety function, $\hgp(\x) \sim \mathcal{GP}(0, k(\x,\x'))$. By using the GP prior, we fully specify the candidate safety function unlike our previous work in \cite{cbf_gaussian_variance_Khan2021}. Only the GP posterior variance was used in \cite{cbf_gaussian_variance_Khan2021} to handle safety uncertainty for a given deterministic CBF. Note that the new formulation is a far more flexible realization \textit{since the data is used for informing both the safety belief and associated uncertainty}. We operate under the following standard assumptions for GPs.

\begin{assumption}\label{assump:noisy_observations}
Each observation $y_i$ is corrupted with Gaussian noise, $y_i \sim \mathcal{N}(p_i, \sigma_y)$, where $p_i$ is the noise-free safety sample and $\sigma_y$ is the observation noise variance. 
\end{assumption}

\begin{assumption}\label{assump:noise_free_training_inputs}
Training input states $\X_N$ in the dataset $\mathcal{D}_N$ are noise-free. 
\end{assumption}

We consider the case of getting noisy measurements for the safety samples. This is a realistic assumption, since in practice these safety samples are captured from noisy sensory measurements. The input training data is considered to be noise-free, however, which is a common practice in machine learning. We first consider the case where the query point $\mathbf{x}_q$ is deterministic, i.e., noise-free. Later, we look at the case when the input query point is also noisy, see Section \ref{subsec:gcbf_noisy_input}.

\begin{assumption}\label{assump:initialize_safeset}
The safe set is nonempty with at least one datapoint, the initial state $\mathbf{x}(0)$ and associated safety value $\hgp(\mathbf{x}(0)) \in \R_{\geq 0}$, to synthesize $\hgp$.
\end{assumption}

We assume that the system begins in an initial compact safe set. Safety for $\hgp$ is encoded as,
\begin{align}
\Sc = \{ \x \in \Rn \ | \ \hgp(\x) \geq 0 \}, \label{eq:gcbf_safeset} \\
\partial \Sc = \{ \x \in \Rn \ | \ \hgp(\x) = 0 \}. \label{eq:gcbf_safeset_boundary}
\end{align}

A noise-free query state $\xquery \in \Rn$ with a noisy safety sample $y \in \R$ is sampled if,
\begin{align}\label{eq:sample_xq}
\norm { \x_{(i)} -  \xquery } \geq \tau, \qquad i = \{1, \cdots, N \},
\end{align}
where $\tau \in \R$ is the sampling distance between any two input states. This avoids dense sampling of the states, resulting in computational tractability. Moreover, if the samples collected are too close together, this may give rise to an ill-conditioned covariance matrix \cite{gp_textbook_Rasmussen2003}. 

\begin{definition}[Gaussian Control Barrier Function]\label{def:gcbf}
\textit{A function $\hgp(\mathbf{x}) : \Rn \rightarrow \R$ is defined as a Gaussian CBF for \eqref{eq:nonlinear_affine_system}, if $\hgp(\mathbf{x}) \sim \mathcal{GP}(0, k(\mathbf{x}_i, \mathbf{x}_j))$ is a Gaussian process, with an infinitely mean-square differentiable positive definite kernel, $k(\mathbf{x}_i, \mathbf{x}_j): \Rn \times \Rn \rightarrow \R$, and if $\exists$ an extended class-$\kappa$ function $\alpha$ such that for any $\x \in \Rn$,}
\begin{align}\label{eq:gcbf_def}
	\sup\limits_{\uu \in \Rm }
	L_f \hgp (\x  ) + 
	L_g \hgp ( \x ) \uu 
	+ \alpha( \hgp(\x )) \geq 0.
\end{align}
\end{definition}

\begin{remark} 
The Gaussian CBF above has attractive properties. A GP prior is placed on the safety candidate function, giving rise to a non-parametric functionality. Thus, the data is used to fully realize the safety function a posteriori. As more data is collected, the overall safety encoded by $\hgp(\x)$ changes. Moreover, it has an analytical form for both the safety belief and uncertainty. This enables computing Lie derivatives of $\hgp(\x)$ in closed-form. 
\end{remark}

The Lie derivatives in (\ref{eq:gcbf_def}) require taking partial derivatives of $\hgp$ with respect to $\x$ which we will discuss later in Section \ref{subsec:lie_gcbf}. We propose the Gaussian CBF $\hgp(\x)$ that incorporates safety belief and uncertainty \textit{online} using the GP posterior mean \eqref{eq:gp_mean} and variance \eqref{eq:gp_var} as follows\footnote{We can employ weights, $w_{\mu}$ and $w_{\sigma^2}$, to the posterior mean and variance respectively in order to adjust safety based on the application. For the sake of simplicity, we consider the weights to be unity in the problem statement.},
\begin{align}
\hgp (\x ) := \mu(\x ) - \sigma^2(\x ) \notag \hspace{4cm}\\
	= \underbrace{\mathbf{k}(\x )^{\top} \KbarInv \y_N}_{\text{safety belief}} 
			- \underbrace{\Big( k (\x , \x ) - \mathbf{k}(\x )^{\top} \KbarInv \mathbf{k}(\x ) \Big)}_{\text{safety uncertainty}}. \label{eq:gcbf}
\end{align}
The GP posterior mean represents the belief we have regarding safety whereas the GP posterior variance accounts for safety uncertainty. We require the following theorem to discuss forward invariance properties for the set $\Sc$ in \eqref{eq:gcbf_safeset}-\eqref{eq:gcbf_safeset_boundary}.

\begin{theorem}[Sample Path Differentiability \cite{gp_nonstationary_proofs_Paciorek2003}]\label{thm:gp_differentiability}
\textit{
A Gaussian process with an isotropic correlation function that can be expressed
in the Schoenberg representation \cite{math_metricspace_schoenberg1938}, has $M^{\text{th}}$-order mean-square partial derivatives if $M$ moments of the length-scale parameter, $l$, are finite.
}
\end{theorem}

We first consider an unforced dynamical system given by $\dot{\mathbf{x}} = f(\mathbf{x})$, where $\uu(t) = 0, \forall t \geq 0$. In this case, the Gaussian CBF will simply be considered as a Gaussian barrier function, since the control input does not appear.

\begin{proposition}\label{prop:gcbf}
\textit{Given a system $\dot{\x} = f(\x)$ with a nonempty safe set $\Sc$ as defined by (\ref{eq:gcbf_safeset}-\ref{eq:gcbf_safeset_boundary}) for a Gaussian process $\hgp$, if $\hgp$ is a Gaussian barrier function defined on the set $\Sc$, then $\Sc$ is forward invariant. 
}
\end{proposition}

\begin{proof}
First, we observe that $\hgp$ uses an infinitely MS differentiable kernel. Hence, $\hgp$ is also infinitely MS differentiable with respect to $\x$ due to Theorem \ref{thm:gp_differentiability} since the length-scale has infinitely many moments. Since $\hgp$ is a Gaussian CBF and infinitely MS differentiable, then the inequality $L_f \hgp (\x) \geq -\alpha( \hgp(\x))$, is satisfied. Given Assumption \eqref{assump:initialize_safeset}, the set $\Sc$ is nonempty, for any $\x \in \partial \Sc$, $\hgp(\x) = 0$ holds. As a result, $\alpha(\hgp(\x)) = 0$ which gives $L_f \hgp(\x) \geq 0 \implies \dot{h}_{\mathrm{gp}} \geq 0$. By applying Nagumo's theorem \cite{controls_sets_Blanchini2008}, which states that for any $C^1$ function $\hgp$, the condition $\dot{h}_{\mathrm{gp}} \geq 0$ on $\partial \Sc$ is necessary and sufficient for  the set $\Sc$ to be forward invariant, completes the proof.
\end{proof}

\begin{remark}
For the case when the kernel is only $M$ times differentiable, and not infinitely MS differentiable, we require that $M > 2\rho$, where $\rho \in \mathbb{N}$ is the relative degree of the system. The proof above holds trivially for an $M$ times differentiable kernel using Theorem \ref{thm:gp_differentiability}.
\end{remark}

We are interested in ensuring forward invariance of $\Sc$ characterized by $\hgp$ for the system defined by ($\ref{eq:nonlinear_affine_system}$). The admissible control space for the Gaussian CBF is given by,
\begin{align}\label{eq:gcbf_inequality}
\mathrm{K}_{\mathrm{gcbf}} \hspace{-0.05cm} = \hspace{-0.1cm} \{ \uu \in \Rm \big| 
L_f \hgp(\x) \hspace{-0.07cm} +  \hspace{-0.07cm}
L_g \hgp(\x) \uu + \hspace{-0.06cm} 
\alpha( \hgp(\x)) \hspace{-0.06cm} \geq \hspace{-0.06cm} 0 \}.
\end{align}

\begin{proposition}\label{prop:gcbf_control}
\textit{Given a Gaussian CBF $\hgp(\x) : \Sc \rightarrow \R$ defined by (\ref{eq:gcbf_def}), where $\Sc$ is nonempty (\ref{eq:gcbf_safeset}), any Lipschitz continuous controller $\uu \in \Rm$, that satisfies (\ref{eq:gcbf_inequality}) for any $\x \in \Rn$, renders $\Sc$ forward invariant for the system (\ref{eq:nonlinear_affine_system}).
}
\end{proposition}

\begin{proof}
$\hgp$ is a Gaussian process with an infinitely MS differentiable kernel. Using Theorem \ref{thm:gp_differentiability}, $\hgp$ is also infinitely MS differentiable and is, therefore, smooth. Since $\hgp$ satisfies (\ref{eq:gcbf_inequality}), we have $L_f \hgp(\x) + L_g \hgp(\x) \uu \geq -\alpha(\hgp(\x))$. Using Theorem \ref{thm:cbf}, the proof is complete.
\end{proof}

\begin{figure}[!t]
\centering
\vspace{0.2cm}
\includegraphics[width=1\linewidth]{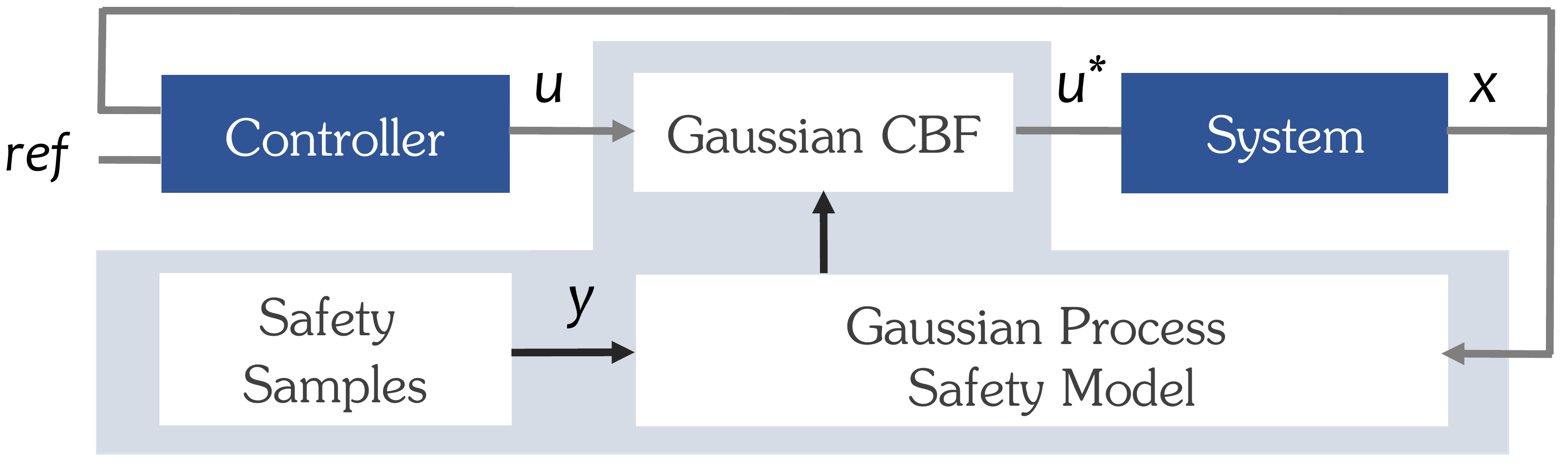}
\caption{Gaussian CBF incorporates safety belief and uncertainty based on past measurement data.}
\vspace{-0.2cm}
\label{fig:gcbf_architecture}
\end{figure}

\subsection{Lie Derivatives of Gaussian CBF}\label{subsec:lie_gcbf}
The Gaussian CBF uses kernels for determining safety belief and uncertainty in the state space. As stated earlier, we use the SE kernel (\ref{eq:gaussian_kernel}) which is an infinitely mean-square differentiable function. Computing its Lie derivatives is necessary for rectifying the control input and ensuring forward invariance for the system in the safe set. First, we take the partial derivative of (\ref{eq:gcbf}) with respect to $\x$ at a query point $\xquery$,
\begin{align}
&\dhgpdx \pipex \hspace{-.3cm} = \dmudx \pipex - \dvardx \pipex \notag \\
& \hspace{1cm} = \y_N^{\top} \KbarInv \dkdx \pipex \hspace{-0.3cm} 
		+ 2 \mathbf{k}( \xquery)^{\top} \KbarInv \dkdx \pipex \hspace{-0.3cm} . \hspace{-0.1cm} \label{eq:dhgpdx}
\end{align}
The kernel derivative in (\ref{eq:dhgpdx}) is given by,
\begin{align}\label{eq:kernel_deriv}
\dkidx \pipex &= (\x_{(i)} -  \xquery)^{\top} \ k(\x_{(i)} ,  \xquery) \mathbf{L}^{-2} ,
\end{align}
where $\mathbf{k}_{(i)}$ is the $i^{\mathrm{th}}$ element of $\mathbf{k}(\x)$, and (\ref{eq:kernel_deriv}) is the $i^{\text{th}}$ row of $\dkdx \in \mathbb{R}^{N \times n}$. Now, we can compute the Lie derivatives of $\hgp(\x)$ by taking its time derivative as follows,
\begin{align}
\dot{h}_{\mathrm{gp}}(\x) &= \dhgpdx f(\x) + \dhgpdx g(\x) \uu		\notag 	\\
	&= L_f \hgp (\x) + L_g \hgp (\x) \uu, \label{eq:gcbf_lie}
\end{align}
\noindent where (\ref{eq:dhgpdx}) is used in the Lie derivatives, $L_f \hgp (\x)$ and $L_g \hgp (\x)$. In prior literature, the derivative predictions of GP posterior mean and variance are exploited \cite{gp_derivatives_observations_Solak2003, gp_derivatives_scaling_Eriksson2018}. 
Here, we take the partial derivatives of $\mu$ and $\sigma^2(\x)$ with respect to the state $\x$ which are different from the derivative predictions of GP posterior mean and variance.

\subsection{Online Safety Control}\label{subsec:online_qp}
Consider a nominal control input $\uu_{\nom} \in \Rm$ that is designed as the feedback policy for system \eqref{eq:nonlinear_affine_system}. This control policy may not restrict the solution of system \eqref{eq:nonlinear_affine_system} inside the safe set. An online quadratic program (QP) rectifies $\uu_{\nom}$ whose constraints are given by the Lie derivatives in (\ref{eq:gcbf_lie}) \cite{cbf_Ames2017}. The QP optimization routine is set up as follows:

\begin{algorithm}
  Gaussian CBF-QP: \textit{Input modification}
	\begin{align}\label{eq:gcbf-qp_degree1}
		 \uu_\mathrm{rect} &= \argmin_{ \uu \in \Rm} \frac{1}{2} \norm{ \uu - \uu_{\nom} }^2 \ \ \text{s.t.} \\
			 & \ \  L_f \hgp (\x) + L_g \hgp (\x) \uu + \alpha( \hgp (\x)) \geq 0, \notag
	\end{align}
\end{algorithm}

\noindent where $\urect$ is the rectified control input. The QP constraint above ensures that the nominal control is followed as long as the safety condition is not violated, i.e., $\hgp(\mathbf{x}) \geq 0$. When approaching the boundary of the safe set, i.e., $\hgp \rightarrow 0$, the QP rectifies $\unom$ minimally to $\urect$. By rectifying the control policy, the system is guaranteed to remain forward invariant for the safe set $\Sc$ due to Proposition \ref{prop:gcbf_control}. When solving for the QP, every term in the constraint is simply a numerical value except for the decision variable, the control input, which is rectified. Therefore, $\hgp(\mathbf{x})$ being highly non-linear and non-convex does not affect finding the rectified control input. The algorithm for computing safe control input from the Gaussian CBF is shown in Algorithm \ref{algo:gcbf_algorithm}.

\begin{algorithm}
\caption{Gaussian CBF Synthesis \& Safe Control}\label{algo:gcbf_algorithm}

\hspace*{\algorithmicindent} \textbf{Input:} \textsc{GP Prior} $\hgp(\x) \sim \mathcal{GP}(0, k(\x,\x'))$ \\
\hspace*{\algorithmicindent} \hspace{1cm} \textsc{System} (\ref{eq:nonlinear_affine_system}) \\
\hspace*{\algorithmicindent} \hspace{1cm} \textsc{Nominal Input} $\uu_{\nom}$

\hspace*{\algorithmicindent} \textbf{Output:} \textsc{Rectified input} $\uu_{\mathrm{rect}}$

\begin{algorithmic}[1]
\Procedure{SafeControl}{}
      	\State \textsc{Sample} $\mathbf{X}_N \gets \xquery \ \& \ \mathbf{y}_N \gets y$ using (\ref{eq:sample_xq})
      	\State \textsc{Synthesize} $\hgp(\X_N, \y_N)$ using (\ref{eq:gcbf})
      	\State \textsc{Compute} $\dhgpdx$ using (\ref{eq:dhgpdx}-\ref{eq:kernel_deriv})
      	\State \textsc{Setup QP} constraint using (\ref{eq:gcbf_lie})
      	\State \textsc{Rectify} $\uu_{\nom}$ using (\ref{eq:gcbf-qp_degree1}) \\
   		\textbf{return} $\uu_\mathrm{rect}$
\EndProcedure
\end{algorithmic}
\end{algorithm}
\vspace{-.25cm}

\begin{remark}
Note that due to the non-parametric nature of the Gaussian CBF, the algorithm above can be treated as a blackbox routine. This is a beneficial property since, if a traditional CBF is altered, then the corresponding Lie derivatives also change explicitly in their form. However, in the Gaussian CBF, the structure of the Lie derivatives remains the same, i.e, the partial derivatives are explicitly agnostic to the underlying CBF. It is characterized only by the data and the dynamical system.
\end{remark}

\subsection{Gaussian CBF with Noisy Query State}\label{subsec:gcbf_noisy_input}
We extend Gaussian CBFs to handle the case when the query state, $\xnoisy$, is stochastic and therefore a random variable.

\begin{assumption}\label{assump:noisy_test_inputs}
The query state $\x_q$ is Gaussian distributed, $\underline{\x}_q \sim \mathcal{N}(\bm{\mu}, \bm{\Sigma})$, where $\bm{\mu}$ is the mean and $\bm{\Sigma}$ is its noise covariance matrix.
\end{assumption}

This has practical significance because accurate estimates of these states are required to generate safe control actions. In practice, however, measurement uncertainty is pervasive leading to error in the state estimates, thus degrading the safety behavior. As a result, we need to modify the posterior predictions of the GP in order to account for this noise.

The predictive equations for a Gaussian test input have been looked at before \cite{gp_multistep_girard2002, gp_deisenroth2011}. Generally, if a Gaussian input is multiplied with the nonlinear GP predictive distribution, the resulting distribution is non-Gaussian,
\begin{align}\label{eq:noisy_input_nongp}
\hspace{-0.15cm} p(\hgp(\xnoisy) | \ \bm{\mu}, \bm{\Sigma}) = 
\int p( \hgp(\xnoisy) | \xnoisy ) p( \xnoisy | \bm{\mu}, \bm{\Sigma} ) \mathrm{d} \xnoisy.
\end{align}

As a result, moment matching is used to derive the posterior predictions. 
To determine the moments of the predictive function value, both the query distribution and the distribution of the function given by the GP are averaged over. For the SE kernel, the posterior mean and variance can be computed for the predictive distribution in \eqref{eq:noisy_input_nongp} in closed-form\footnote{This statement holds true for all kernels, in particular the SE, polynomial, and trigonometric kernels, if the integral of the kernel multiplied with a Gaussian distribution can be solved analytically.} \cite{gp_thesis_deisenroth2010}. By using the law of iterated expectations, the posterior mean with a noisy query point $\xnoisy$ is given as follows \cite{gp_thesis_deisenroth2010},
\begin{align}\label{eq:noisy_gp_mean}
\mu(\xnoisy) &= \mathbf{q}(\xnoisy)^{\top} \ \KbarInv \y_N,
\end{align}
where $\mathbf{q} = [q_i, \dots, q_N ]^{\top} \in \R^N$ with each $q_i$ representing the expected covariance between $\hgp(\xnoisy)$ and $\hgp(\x_i)$,
\begin{align*}
q_i(\x_i, \xnoisy) &:= \int k(\x_i, \xnoisy) \mathcal{N} (\xnoisy | \bm{\mu}, \bm{\Sigma} ) \mathrm{d} \xnoisy \\
&\hspace{-1.2cm}= \sigma_f^2 | \ \bm{\Sigma} \mathbf{L}^{-2} + \mathbf{I}_n  \ |^{\frac{1}{2}} \exp \bigg( \hspace{-0.15cm} - \frac{1}{2} ( \x_i - \bm{\mu} )^{\top} \mathbf{R}^{-1} ( \x_i - \bm{\mu})  \bigg),
\end{align*}
where $\mathbf{R} = \bm{\Sigma} + \mathbf{L}^{2} \in \R^{n \times n}$. It is interesting to note the case for a deterministic query point $\xquery$, where $\bm{\Sigma = 0}$ and $\bm{\mu} = \xquery$. On comparing \eqref{eq:noisy_gp_mean} with \eqref{eq:gp_mean}, the posterior mean for the noisy input results in the same posterior mean for the noise-free input, since $q_i(\x_i, \xquery)$ collapses to $k_i(\x_i, \xquery)$ in \eqref{eq:gaussian_kernel}. Effectively, the noise-free input point is a special case of the noisy posterior prediction with $\mathbf{\Sigma} = \bm{0}$ and $\bm{\mu} = \xquery$.

For details on the derivation of the preditive variance for the noisy test point, see \cite{gp_thesis_deisenroth2010}. Here, we simply state the posterior predictive variance which is as follows,
\begin{align}\label{eq:noisy_gp_var}
\hspace{-0.2cm} \sigma^2(\xnoisy) = \sigma_f^2 - \mathrm{tr}\big( \KbarInv \mathbf{V} \big) + \bm{\beta}^{\top} \big( \mathbf{V}\bm{\beta} - \mathbf{q} \big),
\end{align}
with $\bm{\beta} = \KbarInv \y_N \in \R^N$, and the entries of $\mathbf{V} \in \R^{N \times N}$ are given by,
\begin{align*}
v_{ij} = \frac{k(\x_i, \bm{\mu}) k(\x_j, \bm{\mu})}{ | 2\bm{\Sigma}\mathbf{L}^{-2} + \mathbf{I}_n |^{\frac{1}{2}} } \exp \big( (\mathbf{z}_{ij} - \bm{\mu})^{\top} \mathbf{T} (\mathbf{z}_{ij} - \bm{\mu}) \big),
\end{align*}
where $\mathbf{T} := (\bm{\Sigma} + \frac{1}{2}\mathbf{L}^2)^{-1} \bm{\Sigma} \mathbf{L}^{-2} \in \R^{n \times n}$ and each entry in $\mathbf{z}_{ij} := \frac{1}{2}(\x_i + \x_j) \in \Rn$. As seen in \eqref{eq:noisy_gp_mean} and \eqref{eq:noisy_gp_var}, both the predictive mean and variance explicitly depend on the mean $\bm{\mu}$ and the covariance matrix $\bm{\Sigma}$ of the Gaussian distributed query state $\xnoisy$.

The Gaussian CBF for a noisy query point $\xnoisy$ is given by,
\begin{align}\label{eq:gcbf_noisy}
\hgp(\xnoisy) := \mu(\xnoisy) - \sigma^2(\xnoisy).
\end{align}
To compute the Lie derivatives, the partial derivative of \eqref{eq:gcbf_noisy} with respect to $\bm{\mu}$ is given by,
\begin{align}
\frac{\partial \mu(\xnoisy)}{\partial \xnoisy} \pipemu &= \y_N^{\top} \KbarInv \dqdx \pipemu \label{eq:gcbf_noisy_mean_derivative} \\
\Bigg[ \frac{\partial \sigma^2(\xnoisy)}{\partial \xnoisy} \pipemuk \Bigg]
	&= 	\hspace{-0.2cm} \sum_{i,j = 1}^{N} \hspace{-0.1cm}
	 \Big(	\frac{-1}{\bar{k}_{ij}} + 	\beta_i \beta_j \Big) \frac{\partial v_{ij}}{\partial x_k}  \pipemuk 
	 \hspace{-0.4cm} - \sum_{i = 1}^N \beta_i \frac{\partial q_i}{\partial x_k}  \pipemuk \hspace{-0.1cm}, \label{eq:gcbf_noisy_variance_derivative}
\end{align}
where $k = \{ 1, \dots, n \}$, $x_k$ and $\mu_k$ are the $k^{\mathrm{th}}$ entries of $\xnoisy$ and $\bm{\mu}$ respectively, and $\bar{k}_{ij} = k(\x_i, \x_j) \in \R$. We emphasize that $\mathbf{V}$ is a symmetric matrix which allows combining the summation with the choice of indices in \eqref{eq:gcbf_noisy_variance_derivative}. The Lie derivatives can then be computed using the equations above, similar to \eqref{eq:gcbf_lie}, to achieve safe constrained control by setting up a QP as constructed in \eqref{eq:gcbf-qp}.

\section{APPLICATION TEST CASE: QUADROTOR}\label{sec:application}
To demonstrate the efficacy of our method, we implement our proposed technique on a quadrotor system. Quadrotors pose an interesting and challenging problem due to their inherently unstable nature. 
We run $3$ separate experiments using Gaussian CBF on a quadrotor: (a) safe constrained control for arbitrary safe sets (b) explore the state space safely and synthesize the safe set online, and (c) safe constrained control in presence of noisy states and compare the performance with regular CBFs. 
First, we review quadrotor dynamics followed by safety rectification using the Gaussian CBF for a quadrotor platform.

\subsection{Quadrotor Dynamics}\label{subsec:quadrotor_dynamics}
We consider the position dynamics and attitude dynamics of a quadrotor model evolving in a coordinate-free framework. This framework uses a geometric representation for its attitude given by a rotation matrix $\mathbf{R}$ on $SO(3) := \{ \mathbf{R} \in \R^{3 \times 3} \ | \ \mathbf{R}^{\top}\mathbf{R} = \mathbf{I}_3, \ \det(\mathbf{R}) = 1 \}$. $\mathbf{R}$ represents the rotation from the body-frame to the inertial-frame. The origin of the body-frame is given by the quadrotor's center of mass, denoted by $\mathbf{r} \in \R^3$. 
A quadrotor is an underactuated system since it has $6$ DOF, due to its configuration space being $SE(3) := \R^3 \times SO(3)$, but $4$ control inputs; thrust $F \in \R$ and moments $\mathbf{M} \in \R^3$. The equations of motion are:
\begin{align}
\dot{\mathbf{r}} 	&= \mathbf{v} , 	\notag \\
m\dot{\mathbf{v}}	&= -m g \mathbf{e}_3 + F \mathbf{R} \mathbf{e}_3 ,	\label{eq:dyn2}\\
\dot{\mathbf{R}}	&= \mathbf{R} \mathbf{\Omega}^{\times},	\notag \\
\mathbf{J} \dot{\mathbf{\Omega}}	&= \mathbf{M} - (\mathbf{\Omega}^{\times} \mathbf{J} \mathbf{\Omega})	, \label{eq:dyn4}
\end{align}
where $\mathbf{v} \in \R^3$ is the velocity in the inertial frame, $m \in \R$ is the quadrotor mass, $g \in \R$ is gravity, $\mathbf{e}_3 = [0 \ 0 \ 1]^{\top} \in \R^3$, $\mathbf{\Omega} \in \R^3$ is the body-frame angular velocity, $\mathbf{J} \in \R^{3 \times 3}$ is the inertia matrix, and $(\cdot)^{\times}: \R^3 \rightarrow so(3)$ is the skew-symmetric operator, such that $\forall \ \mathbf{x}, \mathbf{y} \in \R^3, \ \mathbf{x}^{\times}\mathbf{y} = \mathbf{x} \times \mathbf{y}$.

\subsection{Setpoint Generation for Quadrotor}\label{subsec:setpoint_quadrotor}
For achieving safety constrained control of the quadrotor, we first compute setpoints. These setpoints are sent to a Crazyflie 2.1 \cite{crazyflie}, in the form of desired thrust, $F_{\des}$, and desired roll, pitch, yaw angles, $\eta = [\phi_\des, \ \theta_\des, \ \psi_\des]^{\top} \in \R^3$. The Crazyflie is equipped with a fast response low-level onboard controller that can directly track these setpoint commands. More details regarding the hardware experimental setup are covered in Section \ref{subsec:experiment_setup}.

Given a desired trajectory, $\mathbf{r}_\des \in \R^3$, that is twice differentiable, a second-order integrator model is set up,
\begin{align}\label{eq:second_order_integrator}
\underbrace{
\begin{bmatrix}
\dot{\mathbf{r}} \\ \ddot{\mathbf{r}} 
\end{bmatrix}
}_{\dot{\x}}
&=
\underbrace{
\begin{bmatrix}
\mathbf{0} & \mathbf{I} \\
\mathbf{0} & \mathbf{0}
\end{bmatrix}
\begin{bmatrix}
\mathbf{r}  \\ \dot{\mathbf{r}} 
\end{bmatrix}
}_{f(\x)}
+
\underbrace{
\begin{bmatrix}
\mathbf{0} \\
\mathbf{I}
\end{bmatrix}
}_{g(\x)}
\uu,
\end{align}
where $\x = [\mathbf{r}  \ \mathbf{\dot{r}} ]^{\top} \in \R^6$ and $\uu = \ddot{\mathbf{r}}_\des \in \R^3$. The input $\uu$ in (\ref{eq:second_order_integrator}) is rectified using the synthesized Gaussian CBF generating the following rectified setpoints,
\begin{align}
\phi_\rect &= \frac{ (u_{1, \rect} \sin \psi - u_{2, \rect} \cos \psi) } {g}, \label{eq:rectified_roll}\\
\theta_\rect &= \frac{ (u_{1, \rect} \cos\psi + u_{2, \rect} \sin \psi)} {g}, \label{eq:rectified_pitch}\\
F_\rect &=  m(u_{3, \rect} + g), \label{eq:rectified_thrust}
\end{align}
where the desired yaw is assumed to be zero and small angle approximations are made to invert the dynamics in (\ref{eq:dyn2}) for simplicity \cite{cbf_teleoperation_Xu2018, quadrotor_trajectory_generation_control_Mellinger2012}. Next, we discuss the safety rectification of $\uu$ to compute $\uu_\rect$ using the Gaussian CBF.

\subsection{Online Control Rectification}\label{subsec:quadrotor_safety_rectification}
Given a Gaussian CBF expressed in the position space, the relative degree for system (\ref{eq:second_order_integrator}) is $\rho = 2$. The associated Lie derivatives for the Gaussian CBF in (\ref{eq:gcbf}) are,
\begin{align*}
L_f \hgp (\x) 
	&= 	\Big( \nabla \mu(\x) - \nabla \sigma^2(\x) \Big) ^{\top} f(\x),  \hspace{3cm} \\[3pt]
L_f^2\hgp(\x) 
	&= 	f(\x)^{\top} \Big( \mathbf{H}_{\mu} (\x) - \mathbf{H}_{\sigma^2}  \Big) f(\x) \\[-2pt]
  	& \hspace{0.25cm} +	\Big( \nabla \mu(\x) - \nabla \sigma^2(\x) \Big) ^{\top} \cdot \nabla f(\x) \cdot f(\x),	\\[5pt]
L_gL_f\hgp(\x) 
	&= f(\x)^{\top} \Big( \mathbf{H}_{\mu} (\x) - \mathbf{H}_{\sigma^2} (\x) \Big) g(\x) \\
	& \hspace{0.25cm} + \Big( \nabla \mu(\x) - \nabla \sigma^2(\x) \Big)^{\top} \cdot \nabla f (\x) \cdot  g(\x),
\end{align*}
where $\nabla \mu(\x) = \dmudx^{\top}$ and $\nabla \sigma^2(\x) = \dvardx^{\top}$ are the gradients of GP mean and variance in (\ref{eq:dhgpdx}) and $\nabla f(\x) = \frac{\partial f (\x) }{\partial \x}$ is the Jacobian of $f(\x)$. $\mathbf{H}_{\mu}(\x)$ and $\mathbf{H}_{\sigma^2}(\x)$ are the Hessians of GP mean and variance given by,
\begin{align*}
\mathbf{H}_{\mu}(\x) &= \bigg( \hspace{-0.1cm} \sum_i^N a_i \ddkiddx  \bigg), \\[-8pt]
\mathbf{H}_{\sigma^2}(\x) &= 
- 2 \nabla \mathbf{k}(\x) \mathbf{ \overline{K} \hspace{0.1cm} }^{-1} \nabla \mathbf{k}(\x)^{\hspace{-0.05cm}\top} \hspace{-0.1cm} 
- 2 \bigg( \hspace{-0.1cm} \sum_i^N b_i \ddkiddx  \bigg),
\end{align*}
where $a_i$ is the $i^{\mathrm{th}}$ entry of $\y_N^{\top} \KbarInv \in \R^{1 \times N}$, $b_i$ is the $i^{\mathrm{th}}$ entry of $\mathbf{k}(\xquery)^{\top} \KbarInv \in \R^{1 \times N}$, $\nabla \mathbf{k}(\x) = \dkdx^{\top} \in \R^{n \times N}$, and $\ddkiddx$ is the partial derivative of (\ref{eq:kernel_deriv}) with respect to $\x$. For the case when the query state is noisy, $\xnoisy$, we use the predictive mean and variance \eqref{eq:gcbf_noisy} as described in Section \ref{subsec:gcbf_noisy_input}. Similarly, the corresponding partial derivatives are used to compute the Jacobians and Hessians from \eqref{eq:gcbf_noisy_mean_derivative}-\eqref{eq:gcbf_noisy_variance_derivative} to compute the Lie derivatives for the noisy query state. Given the nominal control input $\uu_\nom \in \R^3$ in (\ref{eq:second_order_integrator}), the QP below rectifies $\uu_\nom$ into $\uu_\rect \in \R^3$,
\vspace{-.1cm}
\begin{algorithm}
  Gaussian CBF-QP: \textit{Input modification}
	\begin{align}\label{eq:gcbf-qp}
		 \uu_\rect &= \argmin_{ \uu \in \R^3} \frac{1}{2} \norm{ \uu - \uu_{\nom} }^2 \ \ \text{s.t.} \\
			 & \ \  L_f^2 \hgp (\x) + L_g L_f \hgp (\x) \uu + \mathcal{K}^{\top} \mathcal{H} \geq 0, \notag
	\end{align}
\end{algorithm}

\noindent where $\mathcal{K} = [k_1 \ k_2]^{\top} \hspace{-0.1cm} \in \R^2$ is the coefficient gain vector, and $\mathcal{H} = [L_f \hgp(\x) \ \hgp(\x) ]^{\top} \hspace{-0.1cm}  \in \R^2$ is the Gaussian Lie derivative vector. The rectified input $\uu_{\mathrm{rect}}$ is used to compute the rectified setpoints using \eqref{eq:rectified_roll}, \eqref{eq:rectified_pitch}, \eqref{eq:rectified_thrust} which are then ultimately sent to the quadrotor.

\section{EXPERIMENTAL VERIFICATION}\label{sec:experiments}
In this section, we discuss the implementation of our method on a hardware quadrotor. We test our proposed formulation in three different scenarios. In the first setting, we demonstrate safe constrained control, where the Gaussian CBF is used to formulate the candidate function. These safe sets are arbitrarily designed and are not limited to taking any convex shape. For the second demonstration, we synthesize the safety function online by exploring the state space while avoiding static collisions. The quadrotor performs safe control within the constructed Gaussian CBF. And for the final scenario, we revisit the constrained control problem for a given candidate function, but in the presence of noisy position states. We compare the safe controlled behavior with a regular CBF. All experiments can be seen here: \small{\url{https://youtu.be/HX6uokvCiGk}}\normalsize.

\subsection{Experiment Setup}\label{subsec:experiment_setup}
We use the Crazyflie 2.1 as the hardware quadrotor. State estimation is performed onboard with the help of an external low-cost lighthouse positioning system \cite{crazyflie}. All computations are done remotely on a ground station equipped with an Intel i7-9800X at $4.4 \si{\giga\hertz}$ processor and 16 GB RAM. The $\texttt{crazyflie\_ros}$ API is used to communicate for interprocess communication, subscribing to pose information, and publishing setpoints over the Crazyradio PA USB dongle \cite{crazyflie_ros_hoenig2017}. Positions and velocities are collected at $20 \si{\hertz}$ with a data capacity set to $300$ samples. Gaussian CBF synthesis and rectification routine (\ref{eq:gcbf-qp}) are run on a parallel thread at $50 \si{\hertz}$ where solving the QP takes under $5 \si{\ms}$. Nominal setpoint commands are sent to the Crazyflie at $100 \si{\hertz}$ with the help of a Logitech joystick controller, which acts as the nominal controller in the QP formulation.

\subsection{Scenario A : Safe Control for Arbitrary Safe Sets}
The objective in this scenario is to demonstrate safe constrained control for any given arbitrary safe set using the Gaussian CBF formulation. We assume a high-level observer or planner provides a map from which we can sample (un)safe locations. For instance, take the example of a satellite view for a street or the indoor map of a warehouse unit, where the goal is to navigate an autonomous agent safely and provide safety specifications at the planning phase. The configuration space for position safety in such settings cannot be designed by hand effectively. By using the data driven design of Gaussian CBFs, we can construct safe sets based on the dataset allowing flexible realizations of safe sets based on information from a high-level planner or observer.

In this scenario, we construct a safety map in $2\mathrm{D}$, where the domain is chosen to be $[-0.35, 0.35]$ along each lateral axis, $\mathrm{x}$ and $\mathrm{y}$. We uniformly sample, $N = 200$, $(\mathrm{x},\mathrm{y})$ input points. The safety sample for each input coordinate is drawn from a uniform distribution, $y \sim \mathcal{U}(a,b)$, where $a = -1.0$ and $b = 2.0$ are the lowest and highest values respectively of the distribution. This gives a discrete $2\mathrm{D}$ safety map, where for each of the $200$ $(\mathrm{x},\mathrm{y})$ coordinates, we have an associated target safety sample. The safety maps are synthesized once and do not change during the experiment, so the sampling distance is set to $\tau = 0$. The hyperparameters are arbitrarily chosen to generate arbitrary safety maps: $\mathbf{L} = \mathrm{diag}(0.1, 0.1), \sigma_f = 1, \sigma_y = 0.01$. The Gaussian CBF characterizes the posterior safety map as follows,
\begin{align}\label{eq:gcbf_experiment}
\hgp(\x) := \mu(\x) - 4\sigma^2(\x).
\end{align}

We generate $3$ arbitrary safe sets and run $3$ separate experiments with the quadrotor always starting in the safe set as shown in Figure \ref{fig:hardware_noisefree_exp3}. For each experiment, we plot the flight trajectory of the quadrotor, its initial and final positions, and the $0$-level set of $\hgp$. The plot of $\hgp(t)$ is also shown for each experimental run. First, we point out that the data generates arbitrary non-convex safe sets. The posterior mean represents the safety belief in $\hgp$, whereas the posterior variance quantifies the notion of safety for regions in the state space where we have few or no samples. We can thus generate very safe realizations of candidate CBFs with high probabilistic bounds. The flight trajectory of the quadrotor always remains inside the safe set based on the QP formulation in \eqref{eq:gcbf-qp}. This can be verified by looking at the $\hgp(t)$ for each experiment, which is always non-negative. 

\begin{figure}[!t]
\centering
\vspace{-0.3cm}
\includegraphics[width=0.9\linewidth]{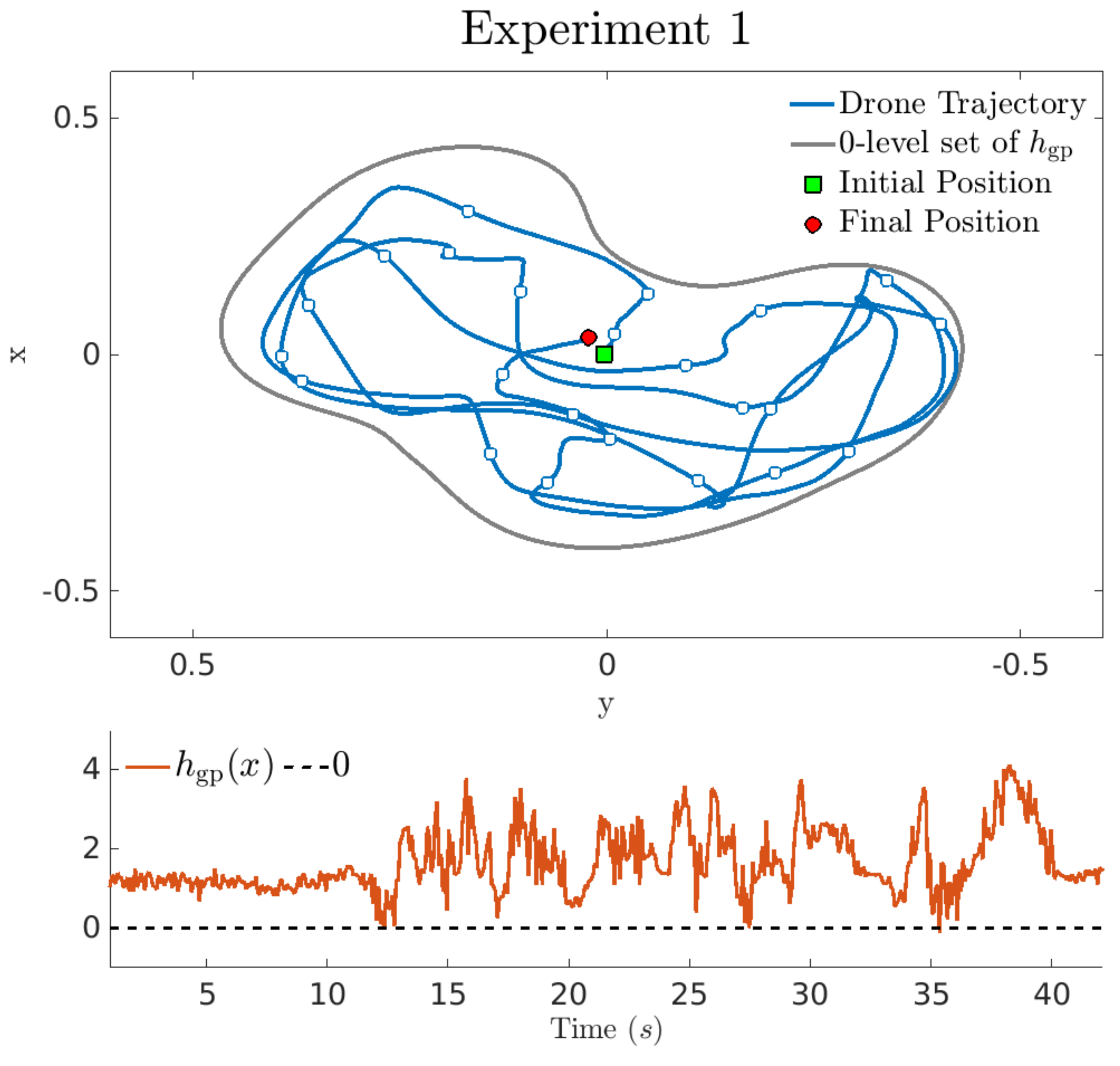}
\includegraphics[width=0.9\linewidth]{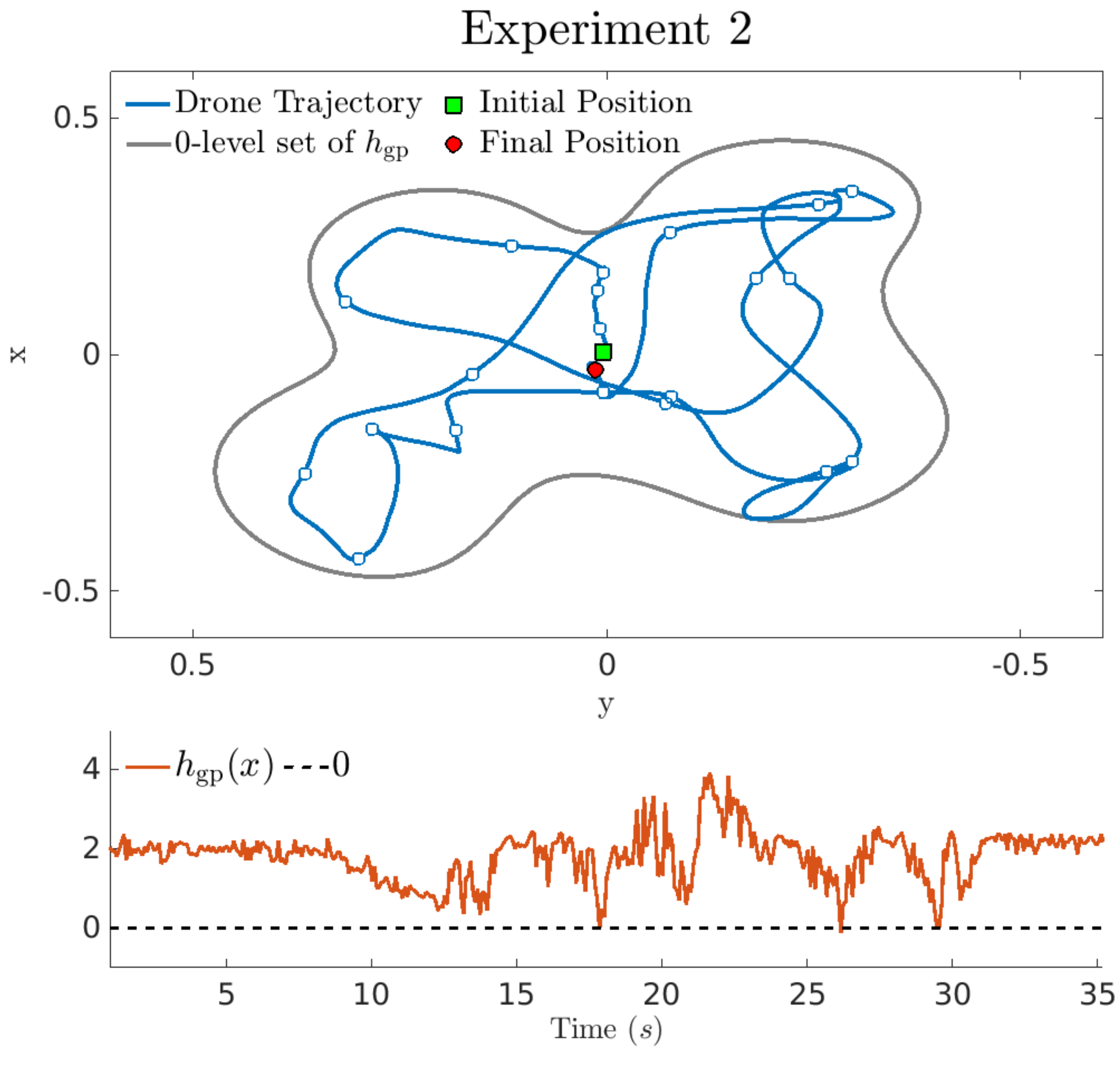}
\includegraphics[width=0.9\linewidth]{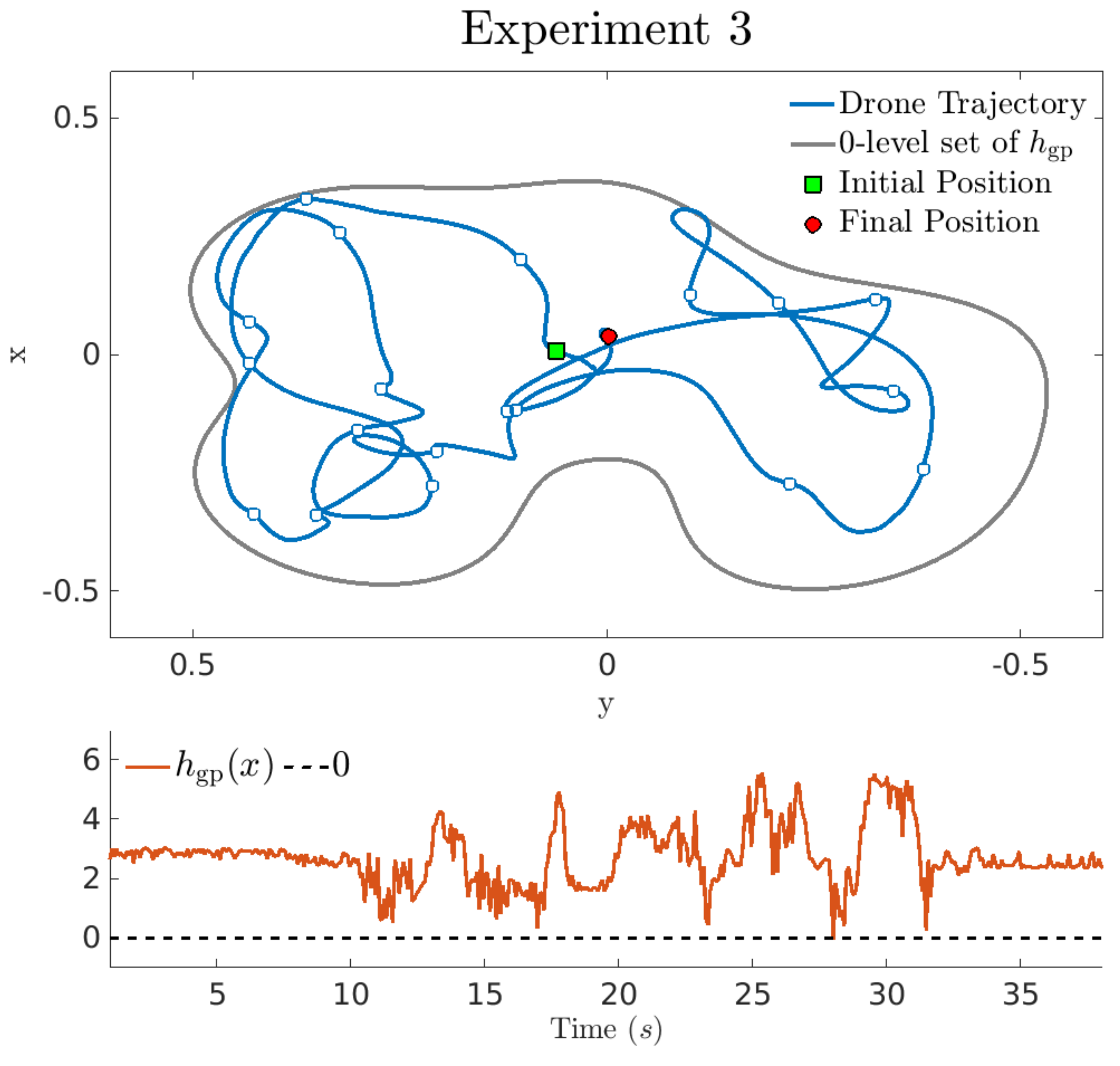}
\vspace{-0.3cm}
\caption{Experiments are run for arbitrary safe sets generated using the Gaussian CBF $\hgp$. The initial (\greensquare), mid-flight (\bluecirc), and final (\reddisc) quadrotor positions are shown. The $0$-level contour line is marked (bold gray). The temporal plot of $\hgp$ for each experiment shows that the quadrotor always remains inside the safe set.}
\label{fig:hardware_noisefree_exp3}
\vspace{-0.3cm}
\end{figure}

\subsection{Scenario B : Online Synthesis of Safe set with obstacle avoidance}

The objective is to synthesize the safety function online by exploring the state space and avoiding collisions. This has great practical significance in safe navigation since onboard sensors are limited in collecting data only within their local proximity. Therefore, we cannot know a priori the complete safety map. Moreover, the sensed data will also need to alter the safety decision boundary online. With the help of Gaussian CBFs, we can incrementally change the safe set as more data is collected. This allows expansion of the safe set in a non-convex manner, which is required in many practical scenarios involving unstructured environments.

In this scenario, $\hgp(\x)$ denotes the distance between the quadrotor and the obstacle. Here, we use two obstacles and therefore have two separate distance measurements,
\begin{align*}
d_a &= (\mathbf{r}_{xy} - \mathbf{p}_a)^{\top} (\mathbf{r}_{xy} - \mathbf{p}_a) - R^2_a,	 	\\
d_b &= (\mathbf{r}_{xy} - \mathbf{p}_b)^{\top} (\mathbf{r}_{xy} - \mathbf{p}_b) - R^2_b	,
\end{align*}
where $\mathbf{r}_{xy}$ is the quadrotor's lateral position, similarly $\mathbf{p}_{(\cdot)} \in \R^2$ is the obstacle's lateral position and $R_{(\cdot)} \in \R$ is the obstacle radius. The overall safety sample is taken as the noisy estimate by combining the two distance measurements,
\begin{align*}
y := d_a \cdot d_b + w, \ \  w \sim \mathcal{N}(0, \sigma_y^2), 
\end{align*}
where $\sigma^2_y \in \R$ is the noise variance. We take noisy sample observations to make the experiment more realistic. Moreover, we also wanted to highlight experimentally that, despite using noisy safety samples, our approach is robust enough to design non-convex safe sets online and ensure the system remains safe. If the quadrotor is closer to one of the obstacles, the product decreases, as a result reducing the safety metric. Note that, even though the obstacles are assumed to be convex, the final safe set constructed need not be convex. This is due to the noisy distance measurements observed and the posteriors being constructed online. The sampling distance is set to $\tau = 0.1$ and the hyperparameters are optimized by maximizing the log marginal likelihood using gradient methods \mbox{\cite{gp_textbook_Rasmussen2003}}. We use the same Gaussian CBF as \mbox{\eqref{eq:gcbf_experiment}} in scenario A to generate the posterior safe set online.

The quadrotor starts in an initial safe set containing only the initial position, see Figure \ref{fig:gcbf_safeset_expansion}. The quadrotor collects safety samples $y$ with the corresponding state $\x$ along its trajectory. With the data being collected, the safety function $\hgp(\x)$ and its associated safe set is constructed online using (\ref{eq:gcbf_experiment}). In regions where we have data, the safety belief is high and safety uncertainty is low. As more data is collected, the associated safe set expands. For unexplored regions in the state space, the safety uncertainty is high due to the high posterior variance. This aligns with the intuition that safety is not known with high confidence in unexplored spaces.

\begin{figure}[!t]
\centering
\vspace{0.1cm}
\includegraphics[width=0.85\linewidth]{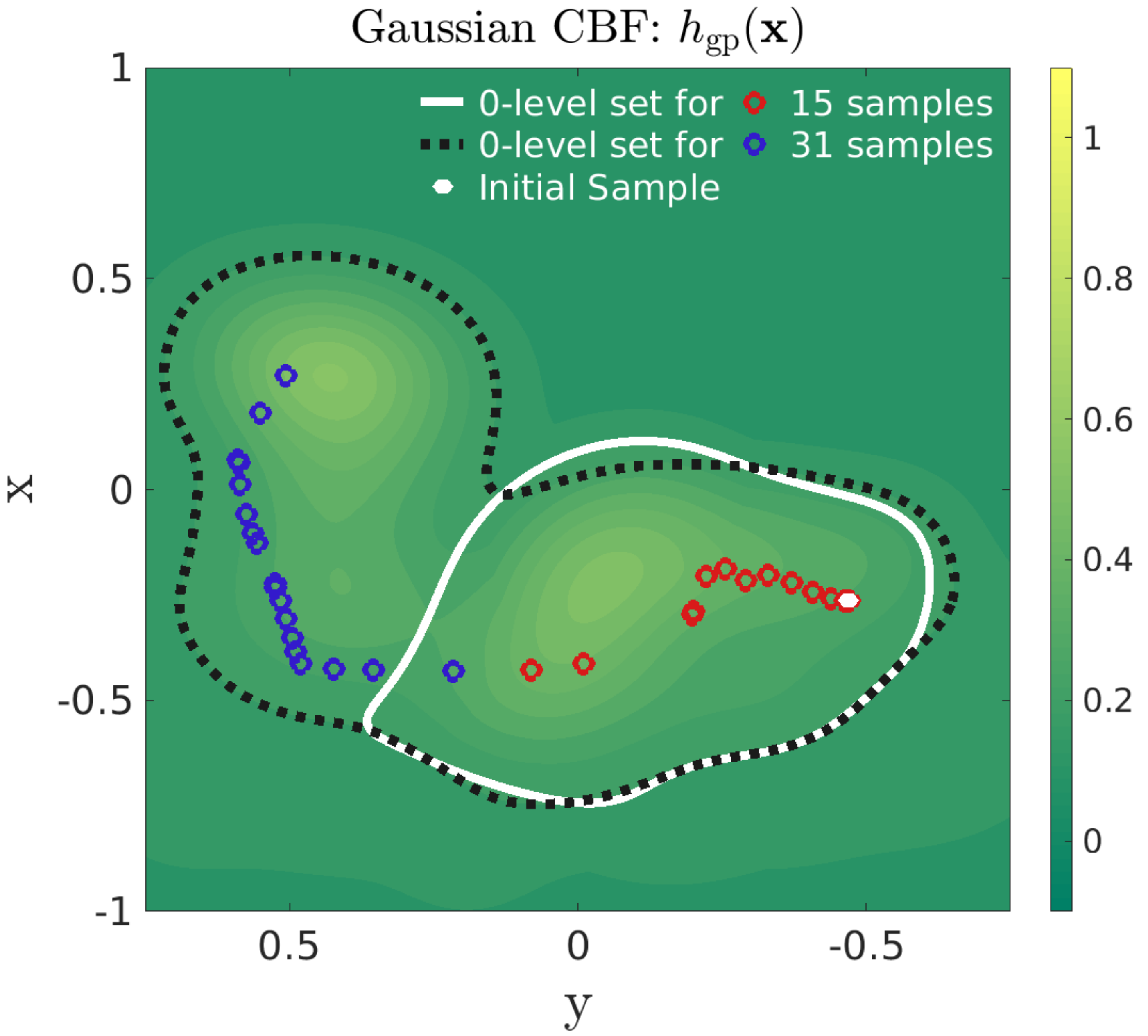}
\vspace{-0.1cm}
\caption{As more samples are collected, the safe set can expand arbitrarily and is not confined to a convex expansion. The contour plots are shown for the two data sample sets. The $0$-level set for $15$ samples is shown with bold white line and for $31$ samples with black dashed line.}
\label{fig:gcbf_safeset_expansion}
\end{figure}

\begin{figure}[!t]
\centering
\includegraphics[width=0.85\linewidth]{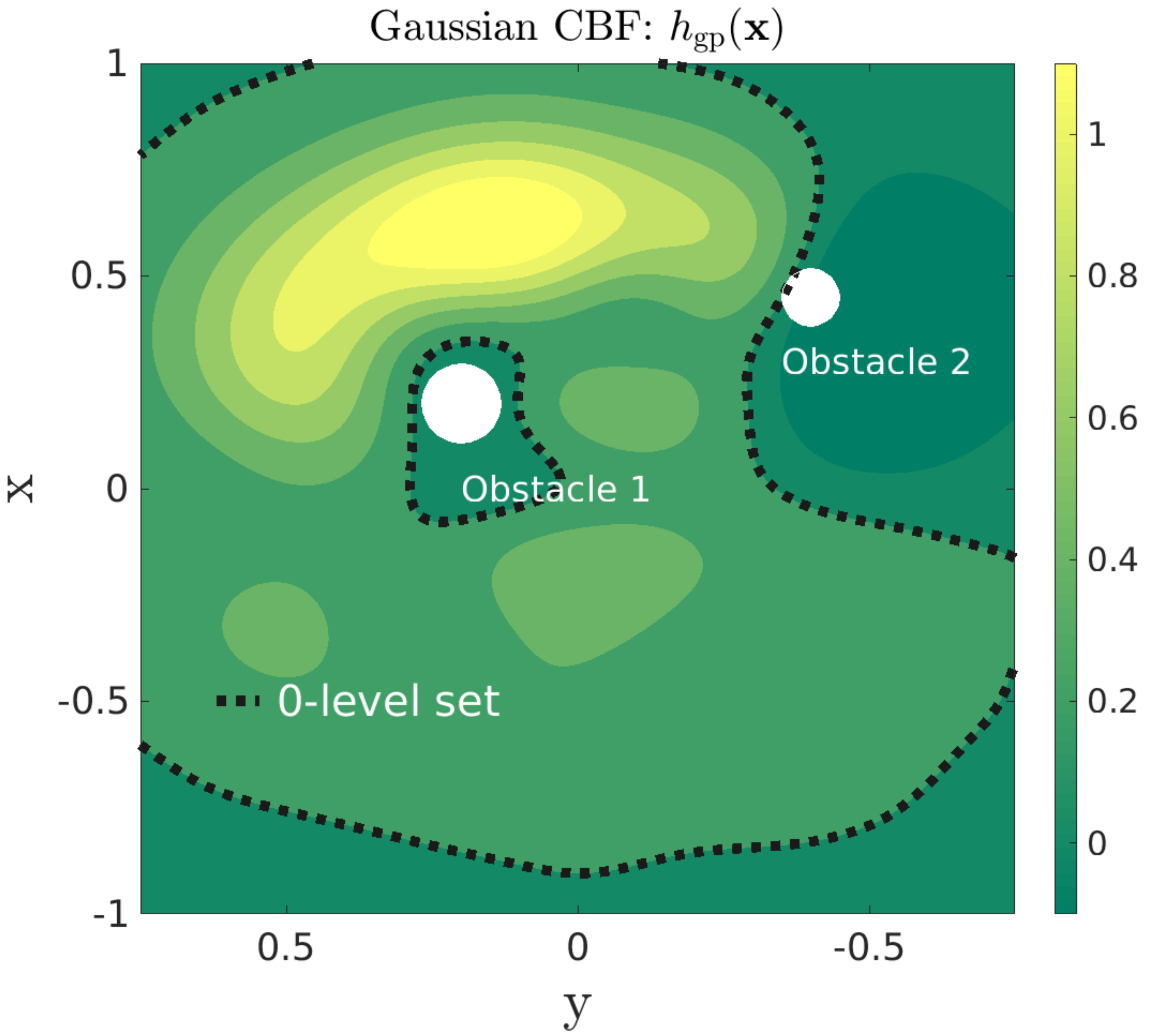}
\vspace{-0.2cm}
\caption{The contour plot for the Gaussian CBF with over $300$ samples collected is shown. After exploring the state space, we see that the obstacles are located in the $0$-sublevel sets. For the collected data set, the $0$-level set is depicted with black dashed line.}
\label{fig:gcbf_safeset_expansion_200_samples}
\vspace{-0.3cm}
\end{figure}

\begin{figure}[!t]
\centering
\includegraphics[width=0.9\linewidth]{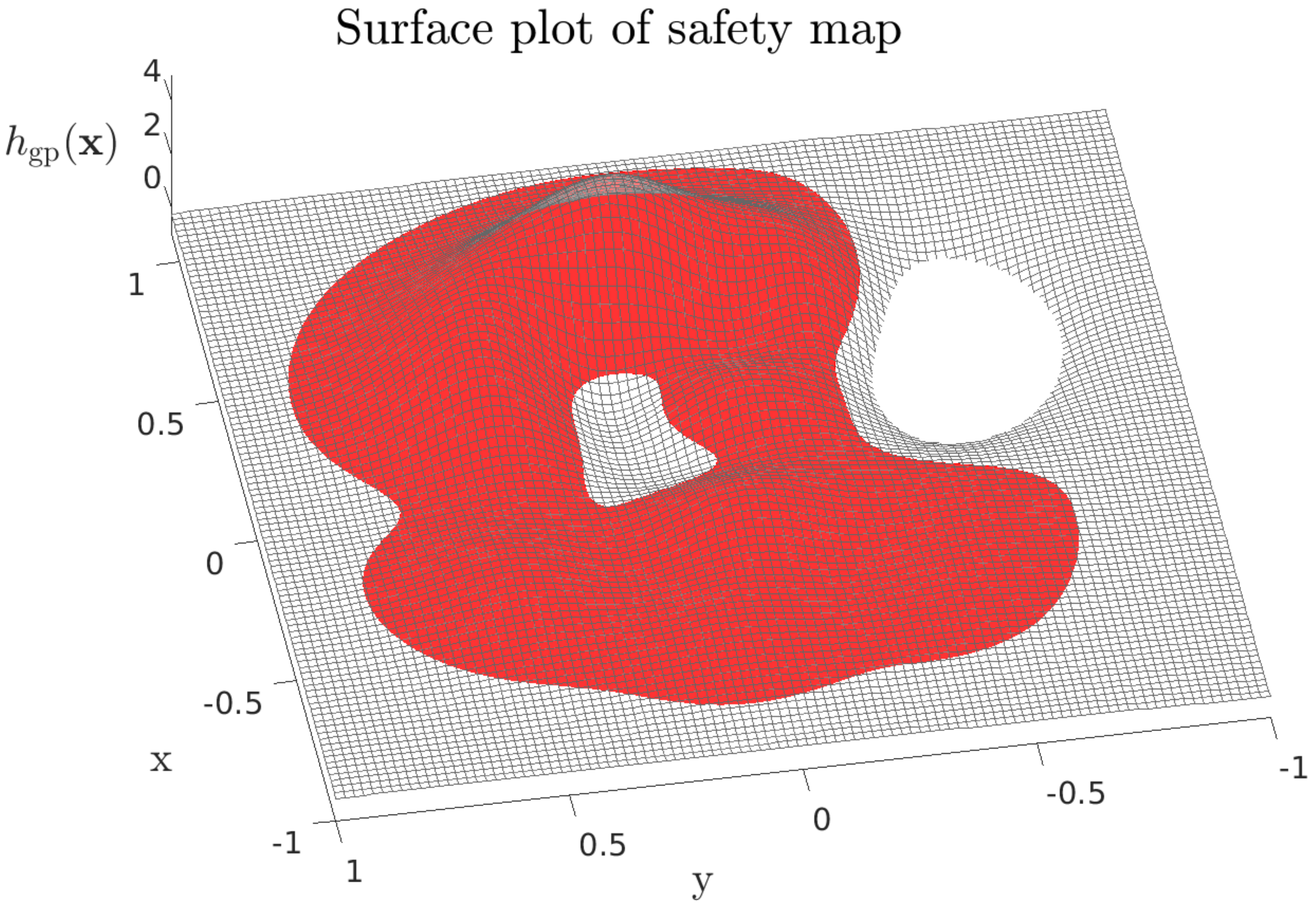}
\caption{Final safe set of $\hgp(\x)$ for the hardware experiment with over $300$ samples collected during the exploration process.}
\label{fig:gcbf_complete_safety_map}
\vspace{-0.4cm}
\end{figure}

\begin{figure}[!b]
\centering
\vspace{-0.3cm}
\includegraphics[width=1\linewidth]{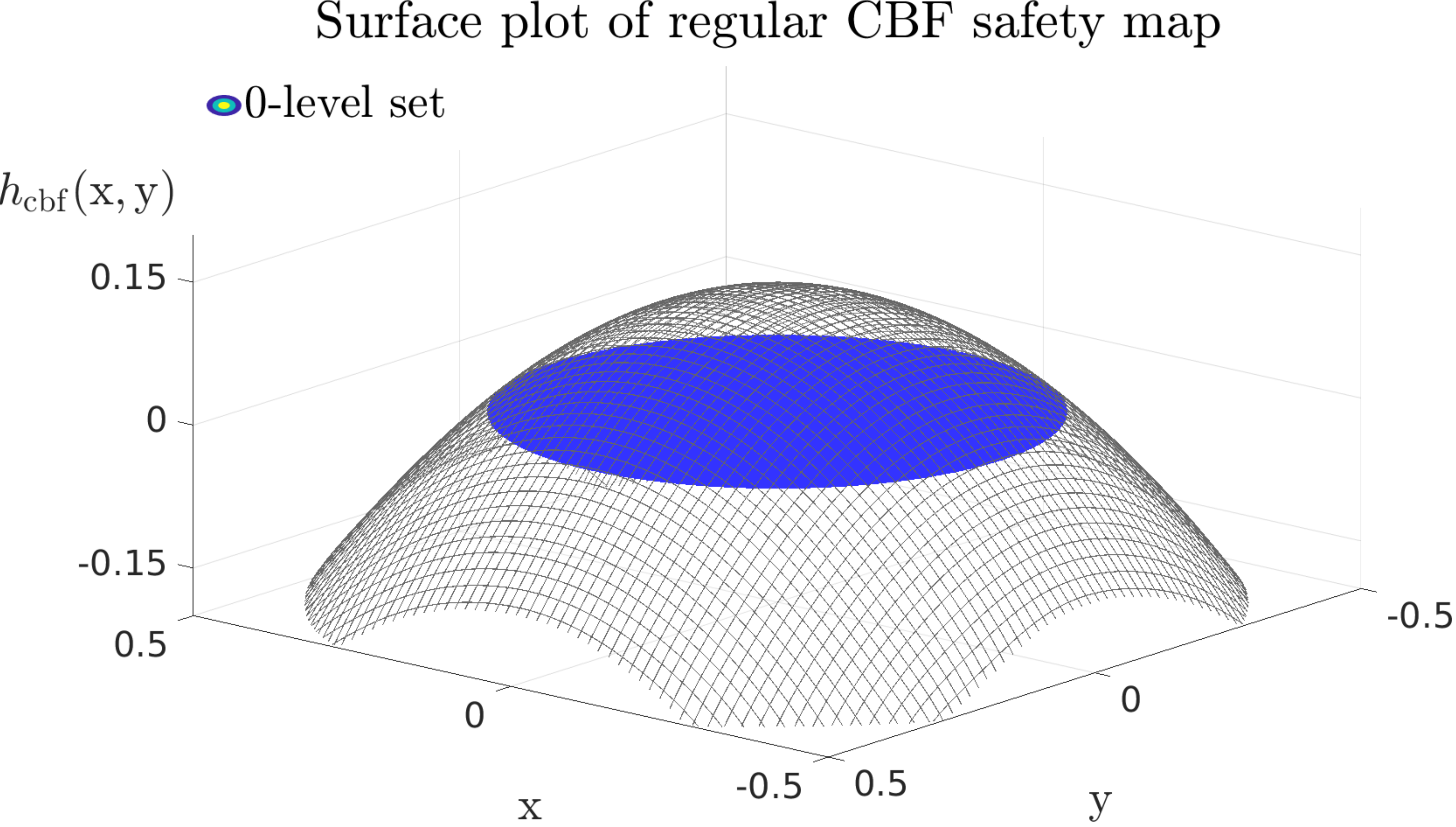}
\includegraphics[width=1\linewidth]{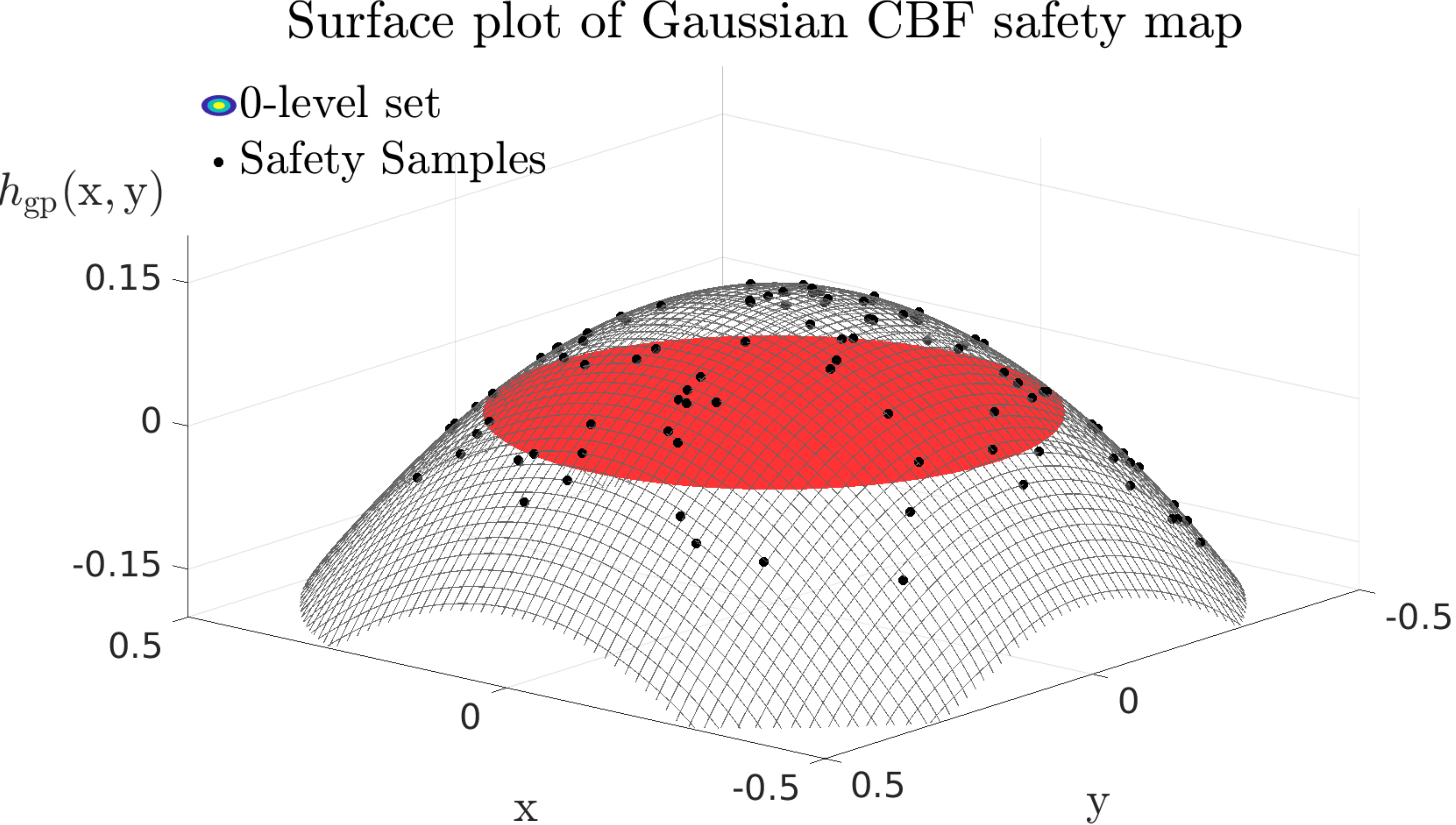}
\caption{Safe sets of CBF $h_{\mathrm{cbf}}(\x)$ (top) and Gaussian CBF $\hgp(\x)$ (bottom) computed by taking $100$ samples ($\bullet$) from $h_{\mathrm{cbf}}$.}
\label{fig:gcbf_safety_map_noisy_experiment}
\vspace{-0.2cm}
\end{figure}

\begin{figure}[!t]
\centering
\vspace{-0.2cm}
\includegraphics[width=0.9\linewidth]{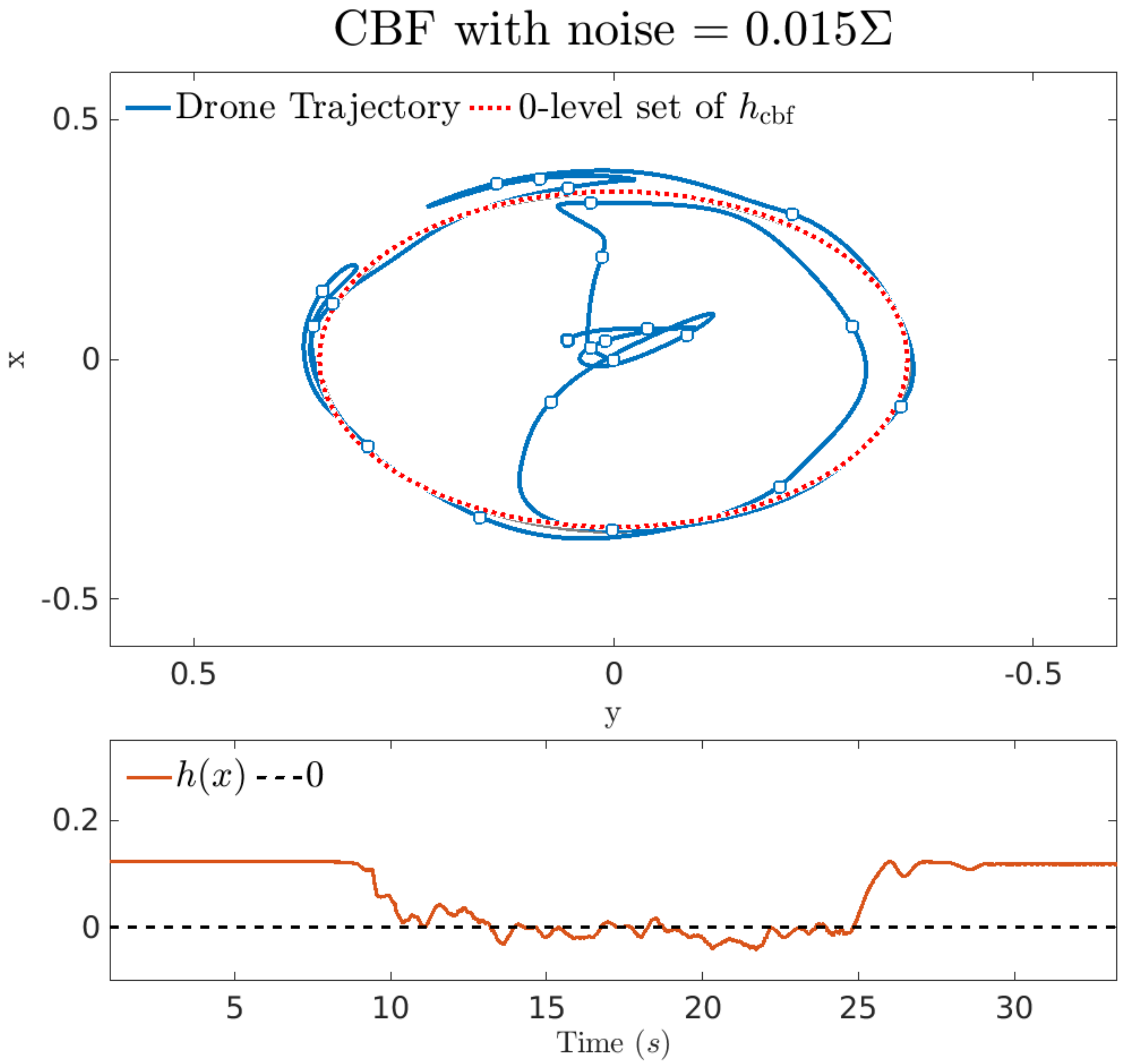}
\includegraphics[width=0.9\linewidth]{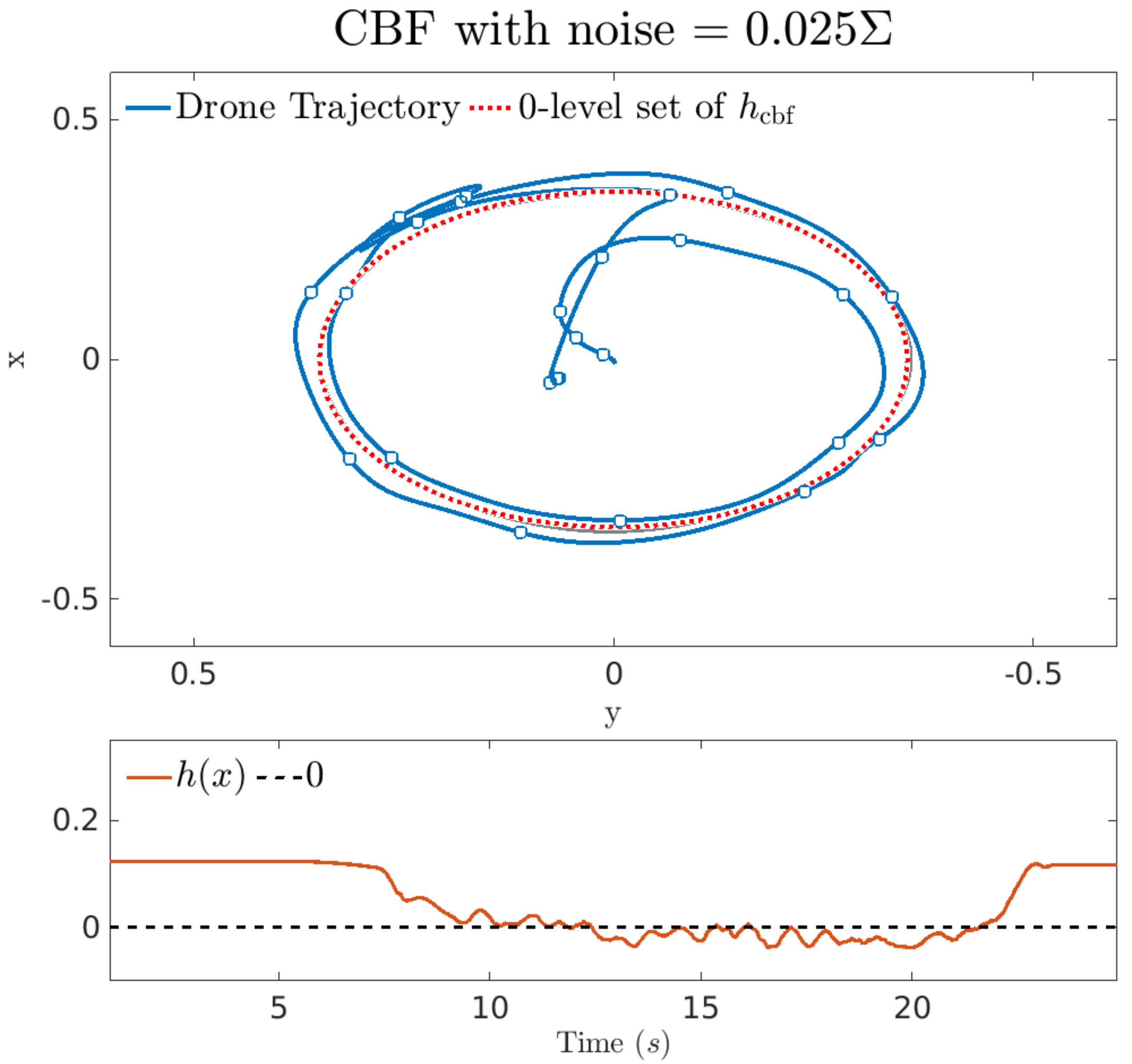}
\includegraphics[width=0.9\linewidth]{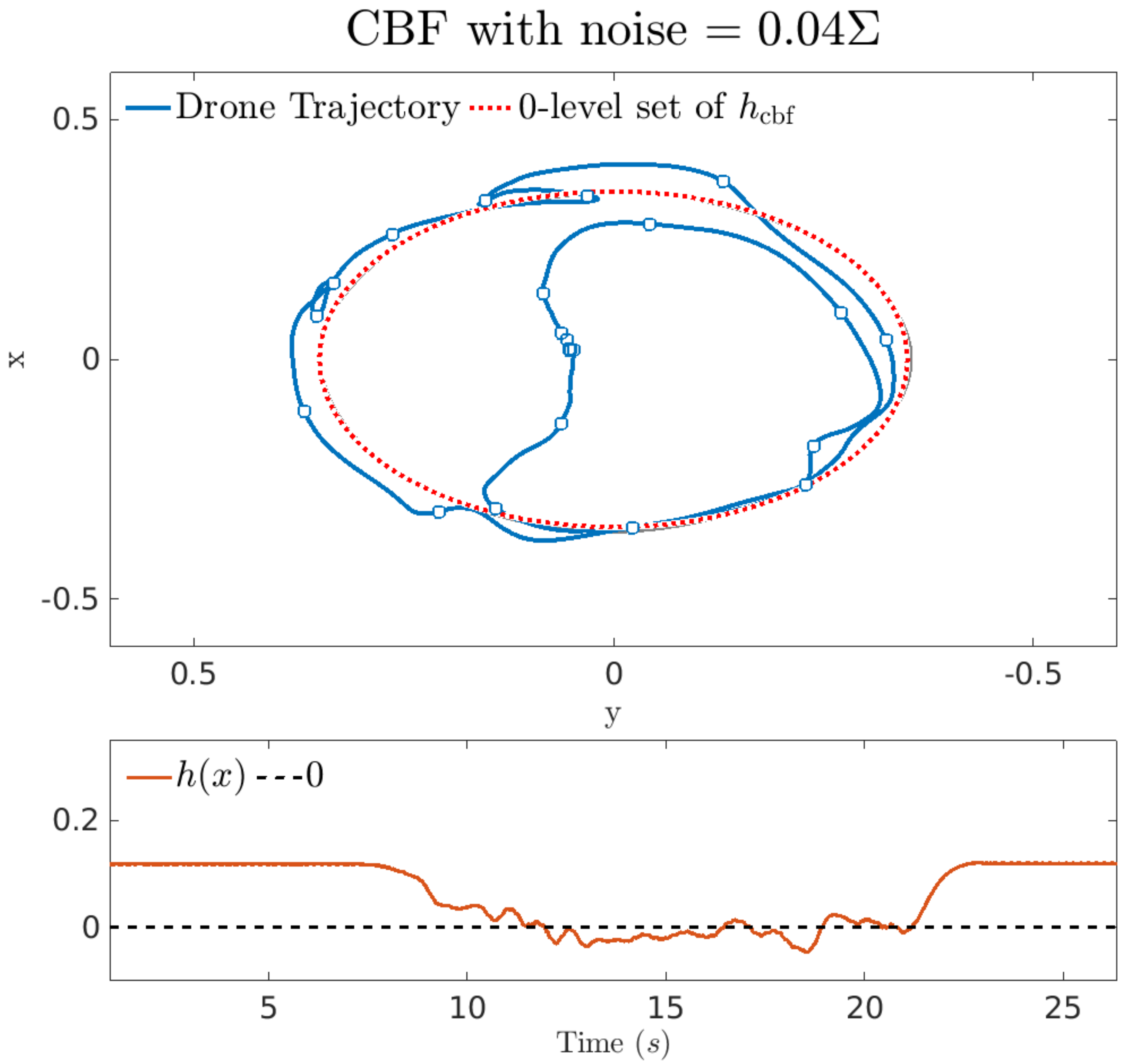}
\vspace{-0.1cm}
\caption{The quadrotor goes outside the boundary of the safe set $R = 0.35$ in presence of noise for CBF. The quadrotor trajectory is shown using ground truth position. The plot of $h_{\mathrm{cbf}}(t)$ using ground truth position data shows that the safety constraint is violated in the presence of noise.}
\label{fig:hardware_noisy04_cbf_exp1}
\end{figure}

\begin{figure}[!t]
\centering
\vspace{-0.2cm}
\includegraphics[width=0.9\linewidth]{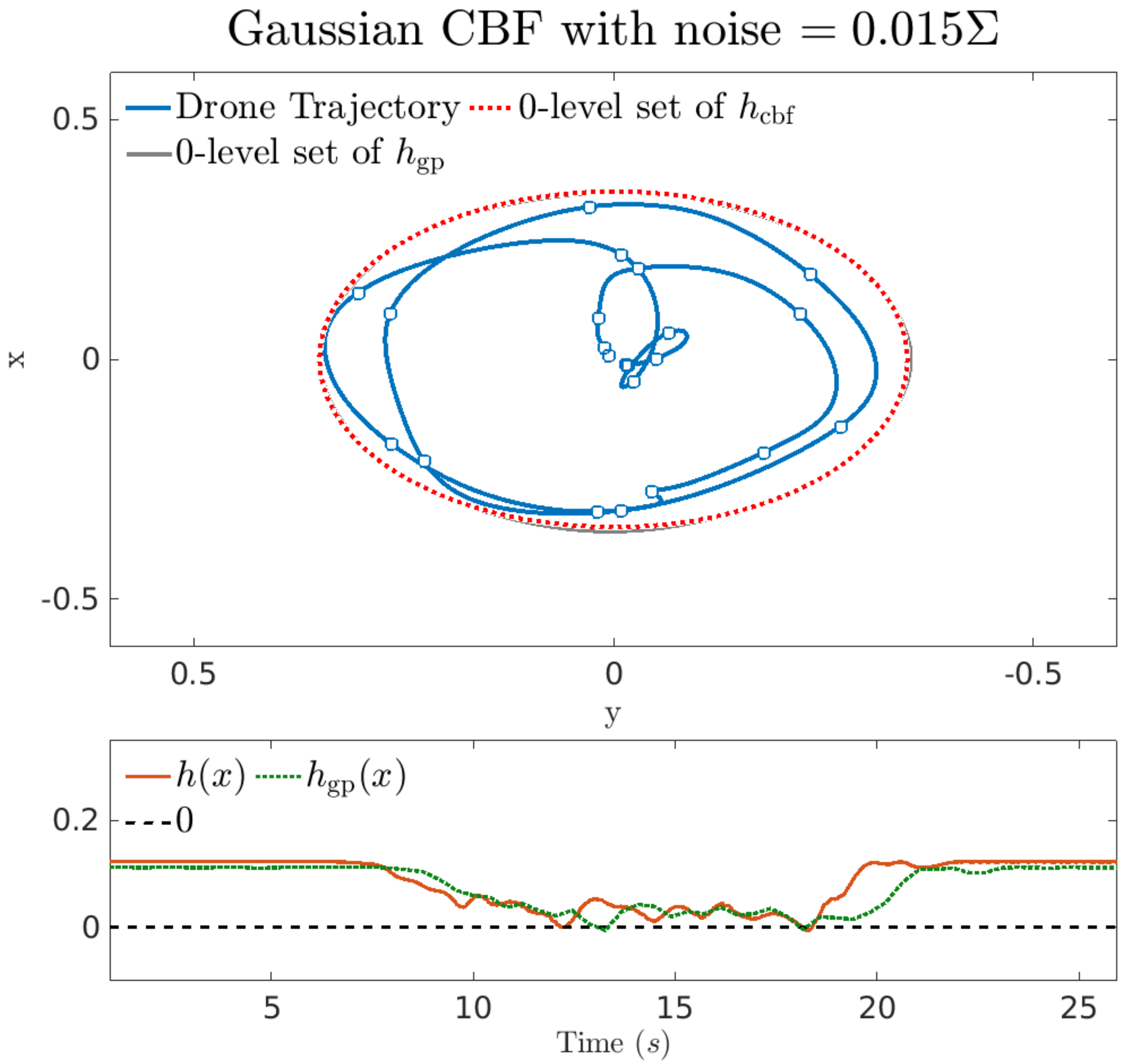}
\includegraphics[width=0.9\linewidth]{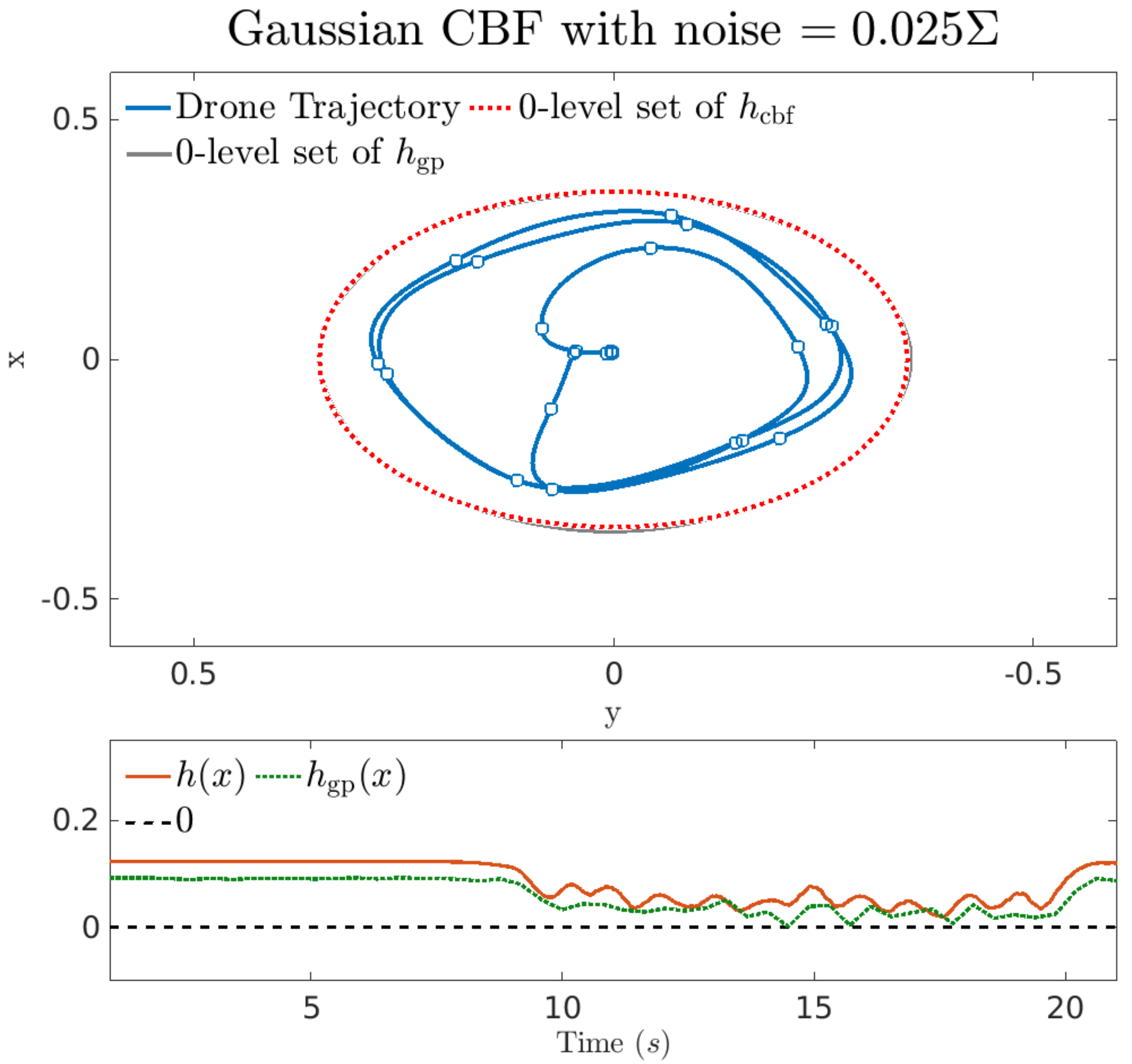}
\includegraphics[width=0.9\linewidth]{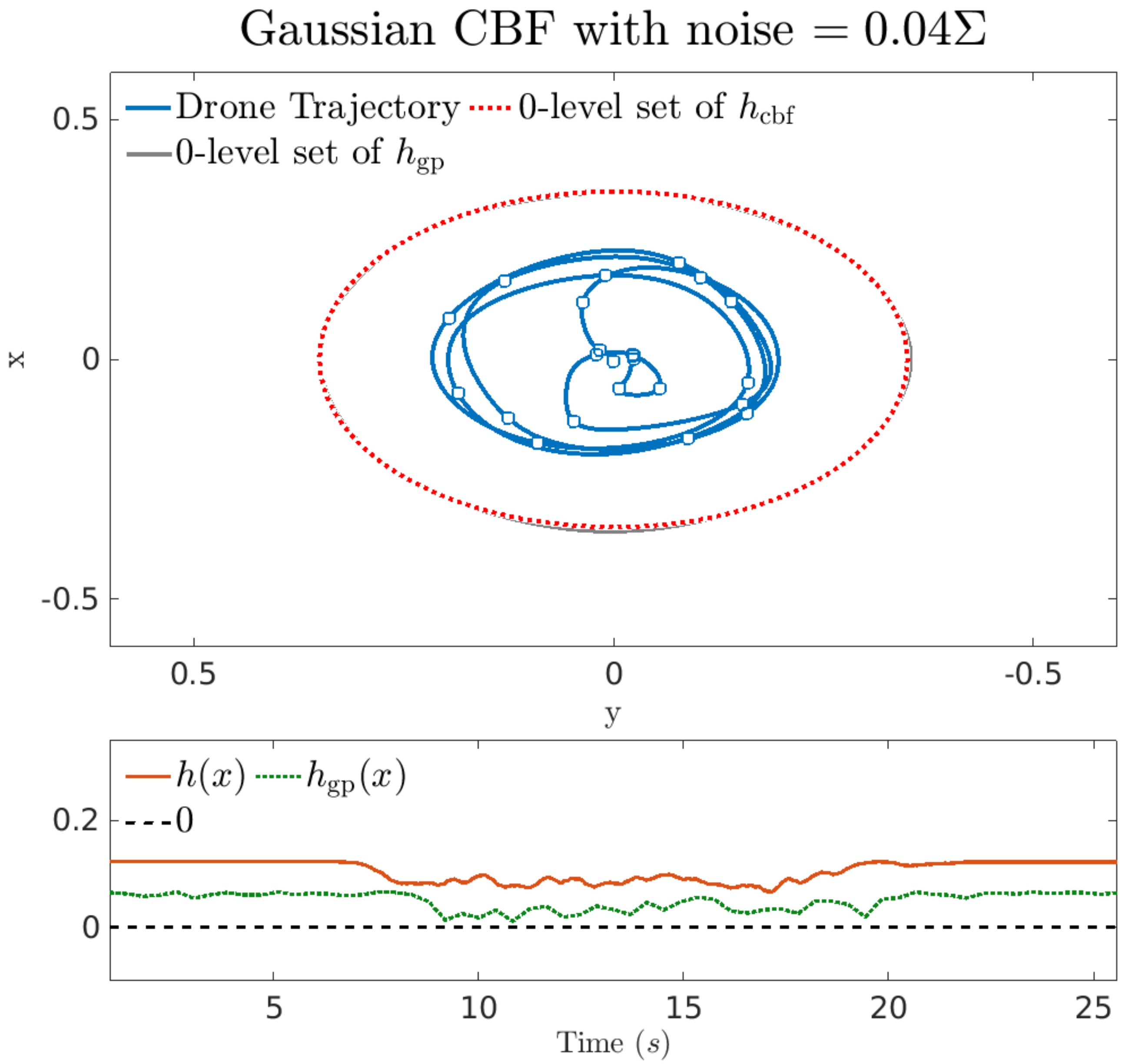}
\vspace{-0.1cm}
\caption{The quadrotor does not go outside the circular boundary of $R = 0.35$ for Gaussian CBF because a more conservative safe set is modeled by $\hgp$. When looking at $\hgp(t)$ and $h_{\mathrm{cbf}}(t)$ on the ground truth position, the quadrotor is further away from the boundary of the safe set using $h_{\mathrm{cbf}}$.
}\label{fig:hardware_noisy04_gcbf_exp1}
\end{figure}

As seen in Figure \ref{fig:gcbf_safeset_expansion}, we see two sets of data samples during the quadrotor's flight. Initially, the quadrotor explored a small region in the state space containing $15$ samples. The $0$-level set is shown in bold white. Then the quadrotor continues exploring the state space further, thus expanding the safe set. Notice that the expansion of the safe set is not limited to any convex expansion. The $0$-level set for the larger safe set of all $31$ samples collected thus far is marked with dashed black line in Figure \ref{fig:gcbf_safeset_expansion}. The exploration process continues, and the quadrotor is able to detect regions in the state space which are unsafe, particularly, when it gets closer to the obstacles. The safe set constructed after collecting more than $200$ samples is shown in Figure \ref{fig:gcbf_safeset_expansion_200_samples}. As seen in the figure, the obstacles are located in the $0$-sublevel sets which are also the unsafe regions. The final safety map for the Gaussian CBF using (\ref{eq:gcbf}) is shown in Figure \ref{fig:gcbf_complete_safety_map}.

\subsection{Scenario C : Safe Control in presence of noisy state}
For the final experiment, we consider the problem of safe constrained control in the presence of noisy position states. In many practical applications, measurement noise is a common occurrence which can degrade system performance and lead to unsafe consequences. This is a particularly hard problem because we consider noise for both the input and observations to the GPs. In our current scenario, this would be noise for the safety samples and the system query state, in particular, the position state of the quadrotor. The query state is given by $\xnoisy := [ \mathbf{\underline{r}}^{\top} \ \mathbf{\dot{r}}^{\top} ]^{\top}$, where $\mathbf{\underline{r}} \sim \mathcal{N}(\mathbf{r}, \bm{\Sigma})$ is the Gaussian distributed noisy position state. We have observability of the noisy query state $\xnoisy$ and assume knowledge of the noise covariance matrix $\bm{\Sigma}$ for the position states.

The safety objective is to keep the quadrotor inside a circle of radius $D$. We first construct a standard CBF using the following candidate function,
\begin{align*}
h_{\mathrm{cbf}} := D^2 - x^2 - y^2,
\end{align*}
where $D = 0.35$ is the safety boundary radius. Next, $100$ samples are sampled randomly with Gaussian noise from this candidate CBF, $y = \mathcal{N}(h_\mathrm{cbf}, 0.03)$. The posterior GP is then computed, along with hyperparameter optimization, using \eqref{eq:gcbf_experiment} from this dataset which forms the Gaussian CBF. Both the CBFs as well as safety samples are shown in Figure \ref{fig:gcbf_safety_map_noisy_experiment}. In this scenario we choose $\tau = 0$, since the safe sets are fixed and require no online sampling.

We perform $3$ separate experiments on the quadrotor for the standard CBF using different values for the noisy position states. The different noise covariance values in the experiments are $\bm{\Sigma} = [ 0.015, \ 0.025 ,\ 0.04] \mathbf{I}_3 \in \R^{3 \times 3}$. For each experiment in Figure \ref{fig:hardware_noisy04_cbf_exp1}, the quadrotor starts inside the safe set and then violates safety at the boundary when subjected to noisy position states. As the value of the noise increases, the quadrotor exhibits more violation of the safety constraint by going outside the safety boundary radius. Since CBFs rectify the control input pointwise and do not account for any measurement noise in its formulation, measurement noise in the position states degrades the safety performance. For each experiment, the temporal behavior of $h_{\mathrm{cbf}}(t)$ verifies that the safety constraint is violated due to the negative values.

We next look at the experiments using Gaussian CBFs for the same values of noise covariance used above. In Figure \ref{fig:hardware_noisy04_gcbf_exp1}, we see that for every experiment, the quadrotor remains confined within a more conservative safe set, which is inside the primary safety boundary radius. This occurs because in the presence of noisy input (query) to the GPs, the posterior mean uses a more conservative weighted kernel $\mathbf{q}$ in \eqref{eq:gcbf_noisy}, whose entries $q_i$ have coefficients lesser than the coefficients of the SE kernel.

The coefficient of an entry $q_i$ given by $\sigma_f^2 | \ \bm{\Sigma} \mathbf{L}^{-2} + \mathbf{I}_n  \ |^{\frac{1}{2}}$ is always lesser than $\sigma_f^2$ since the eigenvalues of $\bm{\Sigma} \mathbf{L}^{-2} + \mathbf{I}_n$ are always greater than $1$ ($ \bm{\Sigma} \mathbf{L}^{-2} $ is a positive definite matrix added to the identity matrix). Therefore, the determinant is always positive and greater than unity. Intuitively, this makes sense since the GP posterior distribution is not overfitting to the noisy input query states, thus leading to a more conservative posterior estimation. We also plot $h_{\mathrm{cbf}}$ as a function of time using the trajectory of the quadrotor rectified under $\hgp$. The quadrotor does not get close to the circular boundary, since it is constrained conservatively by $\hgp$, thus ensuring that the original safety requirement is met. Indeed, if the noise becomes very large, then the safe set may not exist under $\hgp$. Determining the bounds on the measurement noise is currently outside the scope of this study and is left for future investigation. Here, we are primarily interested in achieving safe control in the presence of noisy query states with nonempty compact safe sets.

\subsection{Note on Complexity}
GPs are known to scale cubically with data i.e., $\mathcal{O}(N^3)$, where $N$ is the number of datapoints. 
This complexity arises due to the inverse operation in (\ref{eq:gcbf}) for the covariance matrix $\Kbar$. As the number of data points increases, this could potentially cause a computational bottleneck. We address this with the help of rank-$1$ inverse method. For example, given $500$ samples, it only takes $25 \si{\ms}$ to synthesize the Gaussian CBF. Thereafter, we handle more datapoints by performing rank-$1$ inverse approximations giving tremendous boost in computational speed. For every new data point included, it only takes $4 \si{\ms}$ to compute the inverse covariance matrix.

\section{CONCLUSION}\label{sec:conclusion}
In this study, we proposed a framework for the synthesis of a safety function in a data-driven manner using GPs. The formulation requires safety samples as opposed to the traditional requirement of a smooth function. The newly formulated CBF called the Gaussian CBF was constructed by using a flexible GP prior. GPs provide the posterior mean and variance which serve as analogues for safety belief and uncertainty in our methodology. By exploiting the kernel properties in the posterior mean and variance, we were able to analytically compute the associated Lie derivatives. The Lie derivatives served as constraints in formulating a QP for rectifying the given nominal control input. We empirically verified our framework on a hardware quadrotor platform without risking any expensive system failures. We verify our approach on three different scenarios. The objective in each experiment was to synthesize a candidate safety function using GPs. We successfully show safe control for arbitrary safe sets synthesized using Gaussian CBFs, online synthesis of a Gaussian CBF as more data is collected in a collision avoidance problem, and juxtapose a Gaussian CBF with a regular CBF for constrained control in the presence of noisy position states. The quadrotor always remained inside the safe sets associated with the synthesized Gaussian CBFs. As part of future work, we would like to delve deeper and develop Gaussian CBFs for stochastic settings as well, and compare against other CBF techniques.

\section{ACKNOWLEDGEMENTS}
This research was supported by the US National Science Foundation under Grant S\&AS:1723997.

\bibliographystyle{ieeetr}
\bibliography{root_general22}

\end{document}